\documentclass[11pt,letterpaper]{article}
\usepackage[margin=1in]{geometry}
 \usepackage{authblk}
\usepackage{graphicx,amssymb,amsmath,xcolor}
\usepackage{subcaption}  

\newtheorem{theorem}{Theorem}[section]

\newtheorem{lemma}[theorem]{Lemma}
\newtheorem{remark}[theorem]{Remark}

\newcommand\blfootnote[1]{%
  \begingroup
  \renewcommand\thefootnote{}\footnote{#1}%
  \addtocounter{footnote}{-1}%
  \endgroup
}
\newenvironment{proof}{\textbf{Proof:}}{\hfill$\square$}

\title{Finding Cliques in Geometric  Intersection Graphs  with Grounded or Stabbed Constraints}

\author{J. Mark Keil}
\author{Debajyoti Mondal}
\affil{Department of Computer Science, University of Saskatchewan, Saskatoon,  Canada} 
 
\date{}

\begin{document}

\maketitle
\begin{abstract}
A geometric intersection graph is constructed over a set of geometric objects, where each vertex represents a distinct object and an edge connects two vertices if and only if the corresponding objects intersect. We examine the problem of finding a maximum clique in the intersection graphs of segments and disks under grounded and stabbed constraints. In the grounded setting, all objects lie above a common horizontal line  and touch that line. In the stabbed setting, all objects can be stabbed with a common line. 

\begin{itemize}
    \item We prove that finding a  maximum clique is NP-hard for the intersection graphs of upward rays. This strengthens the previously known NP-hardness for ray graphs and settles the open question for the grounded segment graphs. The hardness result holds in the stabbed setting. 

\item  We show that the problem is polynomial-time solvable for intersection graphs of grounded unit-length segments, but NP-hard for stabbed unit-length segments. 

\item We give a polynomial-time algorithm for the case of grounded  disks. If the grounded constraint is relaxed, then we give an $O(n^3 f(n))$-time $3/2$-approximation for disk intersection graphs with radii in the interval $[1,3]$, where $n$ is the number of disks and $f(n)$ is the time to compute a maximum clique in an $n$-vertex cobipartite graph. This is faster than previously known randomized EPTAS, QPTAS, or 2-approximation algorithms for arbitrary disks. We obtain our result by proving that pairwise intersecting disks with radii in $[1,3]$ are 3-pierceable, which extends the 3-pierceable property from the long known unit disk case to a broader class.
\end{itemize}
\end{abstract}

\section{Introduction}\blfootnote{This work was supported by Natural Sciences and Engineering Research Council of Canada (NSERC) through a Discovery grant and an Alliance grant.}
Let $\mathcal{S}$ be a set of geometric objects. A \emph{geometric intersection graph} $G(\mathcal{S})$ is a graph where each vertex represents a distinct object in $\mathcal{S}$ and an edge connects two vertices if and only if the corresponding objects intersect. In this paper we consider the problem of finding a maximum  \emph{clique} in $G(\mathcal{S})$, i.e., a largest set of pairwise intersecting objects from $\mathcal{S}$.  A rich body of research has examined the maximum clique problem for intersection graphs of various types of geometric objects~\cite{DBLP:conf/focs/BonamyBBCT18,DBLP:journals/dcg/CabelloCL13,DBLP:journals/dam/FelsnerMW97,gavril2000maximum,DBLP:conf/compgeom/KeilM25}. In this paper we consider geometric objects in $\mathbb{R}^2$ and primarily focus on  segments and disks.

\smallskip
\noindent
{\bf Intersection Graphs of Segments.} The maximum clique problem on the segment intersection graph was first posed as an open problem by Kratochv{\'\i}l and Ne{\v{s}}et{\v{r}}il~\cite{kratochvil1990independent} and remained unresolved for two  decades. It was widely regarded as a challenging open problem~\cite{bang2006six} and finally in ESA 2012,  Cabello, Cardinal, and Langerman~\cite{cabello2012clique} settled the problem by proving the NP-hardness for computing a maximum clique in ray intersection graphs. A rich body of research examines generalizations of segment graphs, e.g., \emph{string graphs}, which are   intersection graphs of curves in the plane. In SODA 2011, Fox and Pach~\cite{FoxP11} gave a subexponential time algorithm for computing a maximum clique in $k$-intersecting string graphs. Here a string graph is called \emph{$k$-intersecting} if every pair of  strings have at most $k$ common points. Therefore, the result applies to segment intersection graphs. They also gave a polynomial-time algorithm  to approximate the size of a maximum clique in a $k$-intersecting string graph  within a factor of 
$n^{1-\epsilon_k}$, where  $\epsilon_k>0$ is a parameter that depends on $k$.  They did not provide any detailed derivation for $\epsilon_k$, which relies on the property of specific separators in $k$-intersecting graphs. To the best of our knowledge, this is still the best known approximation for maximum clique in segment intersection graphs. 

\smallskip
\noindent
{\bf Disk Intersection Graphs.} In 1990,  
Clark, Colbourn and Johnson~\cite{DBLP:journals/dm/ClarkCJ90} showed that the maximum clique problem is polynomial-time solvable in the case of unit disk graphs. Subsequent research attempted to accelerate the algorithm~\cite{breu,DBLP:journals/dcg/EppsteinE94,DBLP:conf/compgeom/EspenantKM23}; however, extending this result to disk graphs with even two distinct radii types remained open for a decade~\cite{c1,c2}. In SoCG 2025, Keil and Mondal~\cite{DBLP:conf/compgeom/KeilM25}  settled this open problem showing that a maximum clique can be computed in polynomial time for disk intersection graphs with a fixed number of radii types.   However, the time-complexity question remains open for general disk graphs~\cite{DBLP:conf/waoa/Fishkin03,ambuhl2005clique}, even when the disk radii are restricted to lie within the interval $[1,1 + \epsilon]$ for some  $\epsilon>0$~\cite{DBLP:conf/compgeom/BonnetG0RS18}. The maximum clique problem is NP-hard for the intersection graph of ellipses~\cite{ambuhl2005clique}, where  the ratio of the larger over the smaller radius is some prescribed number larger than 1. However, this hardness result  applies neither to the segments nor to the disks.


\smallskip
\noindent
{\bf Grounded and Stabbed Settings.} In this paper we consider the maximum clique problem for segment and disk intersection graphs under grounded and stabbed settings. In the grounded setting, all objects lie above a common horizontal line $\ell$ and touch  $\ell$. In the stabbed setting, all objects can be stabbed with a common line. Such restrictions have been widely studied in the literature across various problems  with the goals of understanding the structure of specific subclasses~\cite{cabello2017refining,Jelinek019,chakraborty2024recognizing,DBLP:journals/jgaa/CardinalFMTV18}, designing efficient approximation algorithms for the general case~\cite{chan2018stabbing,agarwal2006independent}, and understanding the gap between polynomial-time solvable and NP-hard variants~\cite{DBLP:journals/comgeo/BoseCKM0MS22,bandyapadhyay2019approximating,keil2017algorithm,liu2023geometric}.  

A pioneering result on the maximum clique problem in intersection graphs of grounded segments was established by Middendorf and Pfeiffer~\cite{MiddendorfP92}. They  showed that the problem is polynomial-time solvable for  
grounded segments if the \emph{free endpoints}, which are not on the ground line,  lie on a fixed number of horizontal lines. They also showed that the problem is NP-hard for intersection graphs of one-bend   polylines. Later, Keil, Mondal and Moradi~\cite{DBLP:journals/corr/abs-2107-05198} established NP-hardness for grounded one-bend polylines, even when their free endpoints and bends are restricted to lie on two horizontal lines.  They posed the open question of whether the case of grounded segments is polynomial-time solvable, which was later asked even under the stricter condition where the grounded  segments are of unit length~\cite{DBLP:conf/cccg/cccg2024}. The NP-hardness result for grounded one-bend polylines naturally applies to the stabbed setting. Similarly, the NP-hardness for ray intersection graphs~\cite{cabello2012clique} can be seen as the case  where  all segments are grounded on a circle or stabbed by a circle.

To the best of our knowledge, the intersection graphs of grounded disks have remained unexplored, but several studies have considered various forms of stabbing. For example, Breu\cite{breu} studied unit disk graphs with the restriction when all disks lie within a strip of width $\sqrt{3}/2$, and here the disks can be stabbed by the bisector line of the strip. Another widely studied setting is point stabbing or piercing. Every set of mutually intersecting disks can be  pierced   by four points~\cite{carmi2023stabbing,danzer1986losung,stacho1981solution}, and three points suffice for unit disks~\cite{biniaz2023simple,hadwiger1955ausgewahle}. This idea is useful for designing a 2-approximation in $O(n^4\cdot  f(n))$ time~\cite{ambuhl2005clique} for $n$ disks as follows. First, compute the arrangement $\mathcal{A}$ of the disks, and then compute for each pair of cells $c,c'\in \mathcal{A}$, a maximum clique in the intersection graph $G(c,c')$ of the disks hit by $c,c'$. Here $f(n)$ is a polynomial in $n$, which represents the time to compute a maximum clique in the cobipartite graph $G(c,c')$. A subexponential-time QPTAS and a randomized EPTAS were developed in FOCS 2018~\cite{DBLP:conf/focs/BonamyBBCT18}, but both are   slower than the straightforward 2-approximation.


\smallskip
\noindent
{\bf Contributions.}
In this paper we show that finding a maximum clique is NP-hard in the intersection graphs of  upward rays. This implies NP-hardness both for grounded segments and for unit-length segments that are stabbed by a line, which settles the time-complexity question posed in~\cite{DBLP:journals/corr/abs-2107-05198}. The NP-hardness reduction is achieved by showing that every planar graph $G$ has an even subdivision whose complement can be represented as an upward ray graph, whereas the previous NP-hardness result for ray graphs~\cite{DBLP:journals/dcg/CabelloCL13} used both upward and downward rays to find a geometric representation for $G$. The construction of rays in~\cite{DBLP:journals/dcg/CabelloCL13} relies on a symmetric ray arrangement around a circle, making it unclear how to adapt the method to enforce upwardness of the rays. Instead, we give an incremental construction in which each newly inserted ray is carefully rotated so that all rays remain upward.

In contrast, we give polynomial-time algorithms for two restricted cases. One is for grounded unit-length segments, which settles the time-complexity question posed in~\cite{DBLP:conf/cccg/cccg2024}. Here we reduce the problem to   computing a maximum clique in a cocomparability graph, which is known to have a polynomial-time solution~\cite{gavril2000maximum}. The other is for grounded segments with the free endpoints lying on $O(1)$ lines, which strengthens the result of~\cite{MiddendorfP92} that assumes free endpoints to be on horizontal lines. Since we allow lines that are not necessarily horizontal, the dynamic programming approach of~\cite{MiddendorfP92}, which relies on subproblem decomposition using the heights of the horizontal lines, no longer applies. Our algorithm also employs a dynamic program, but to handle lines of arbitrary slopes, we  prove that only a few  segments per line need to be tracked when performing the subproblem decomposition. For intersection graphs of arbitrary grounded segments, we give an $O(n^{3/4})$-approximation for the size of a maximum clique in polynomial time  using a simpler approach than that of~\cite{FoxP11}.

For grounded disks, we give a polynomial-time solution to  compute a maximum clique. If the grounded constraint is relaxed, then we compute a $3/2$-approximation for $[1,3]$-disks in $O(n^3\cdot f(n))$-time. Here  $f(n)$ is the time to find a maximum clique in an $n$-vertex cobipartite graph. Although we impose restriction on radii, the running time is faster than the previously known 2-approximation algorithm~\cite{ambuhl2005clique}, QPTAS, and randomized EPTAS~\cite{DBLP:journals/jacm/BonamyBBCGKRST21} that work for arbitrary disks. However, the central  contribution here is our underlying proof that every set of pairwise intersecting $[1,3]$-disks can be pierced by 3 points. This extends the 3‑pierceable property, which was long known for unit disks~\cite{biniaz2023simple,hadwiger1955ausgewahle}, to the broader class of [1,3]-disks.

The following table gives a summary of the contribution.

\begin{table}[h]
\centering
\begin{tabular}{|l|c|c|c|c|}
\hline
\shortstack[l]{Objects $\rightarrow$\\ Setting $\downarrow$} &
\shortstack{Upward Rays/\\ Segments} &
\shortstack{Unit-\\Length\\ Segments} &
\shortstack{Segments with \\Endpoints on\\ $O(1)$ lines } &
\shortstack{Disks\\ \phantom{-}\\ \phantom{-}} \\
\hline

\shortstack{Grounded\\ \phantom{-}\\ \phantom{-}} &
\shortstack{NP-hard -- Th.~\ref{thm:uprayhard}\\  $O(n^{3/4})$-approx. \\  Th.~\ref{thm:approx}
} &
\shortstack{Poly-time\\ Th.~\ref{thm:unit}} &
\shortstack{Poly-time\\ Th.~\ref{thm:fixedlines}} &
\shortstack{Poly-time -- Th.~\ref{thm:polydisk}\\ \phantom{-}\\ \phantom{-}} \\ \hline

\shortstack[l]{Stabbed\\ by Line} &
\shortstack{NP-hard \\ Th.~\ref{thm:uprayhard}} &
\shortstack{NP-hard\\ Th.~\ref{thm:uprayhard}} &
Open &
\shortstack{  
 3/2-approx. for [1,3]-disks \\in $O(n^3f(n))$-time -- Th.~\ref{thm:approxdisk}} \\ \hline
\end{tabular}
\end{table}

\section{NP-Hardness for Intersection Graphs of Upward Rays, Grounded Segments, and Stabbed Unit-Length Segments}

In this section we show that the maximum clique problem is NP-hard for intersection graphs of upward rays. Every graph admitting  an upward ray representation also has a grounded segment or stabbed unit-length segment representation, as follows.  Given an upward ray representation, consider a horizontal line $L$ with sufficiently high $y$-coordinate such that all intersection points of the rays are below $L$, remove the parts of the rays that are above $L$, flip the representation vertically, and treat $L$ as the ground line. To obtain a unit-length segment representation, we extend the segments below $L$ so that all segments obtain the same length. Since all intersection points of the rays were on the other side of $L$, the extended parts do not create any new intersection point. Consequently, our NP-hardness result implies NP-hardness for intersection graphs of grounded segments and stabbed unit-length segments

The proof idea is to show that every planar graph has an even subdivision whose complement has an intersection representation with upward rays. Since finding an independent set in an even subdivision of a planar graph is NP-hard~\cite{DBLP:journals/dcg/CabelloCL13}, we obtain the NP-hardness for the maximum clique problem on upward ray graphs. 
We first briefly describe the required subdivisions of planar graphs (used in ~\cite{DBLP:journals/dcg/CabelloCL13}), and then prove that they are intersection graphs of upward rays.  When  constructing the upward ray representation, we do not enforce polynomially bounded coordinates.  
Therefore, our NP-hardness result holds under the assumption that the input is given as a 
combinatorial representation of upward ray graphs, rather than a geometric representation.

\begin{figure}[h]
    \centering
    \includegraphics[width=\linewidth]{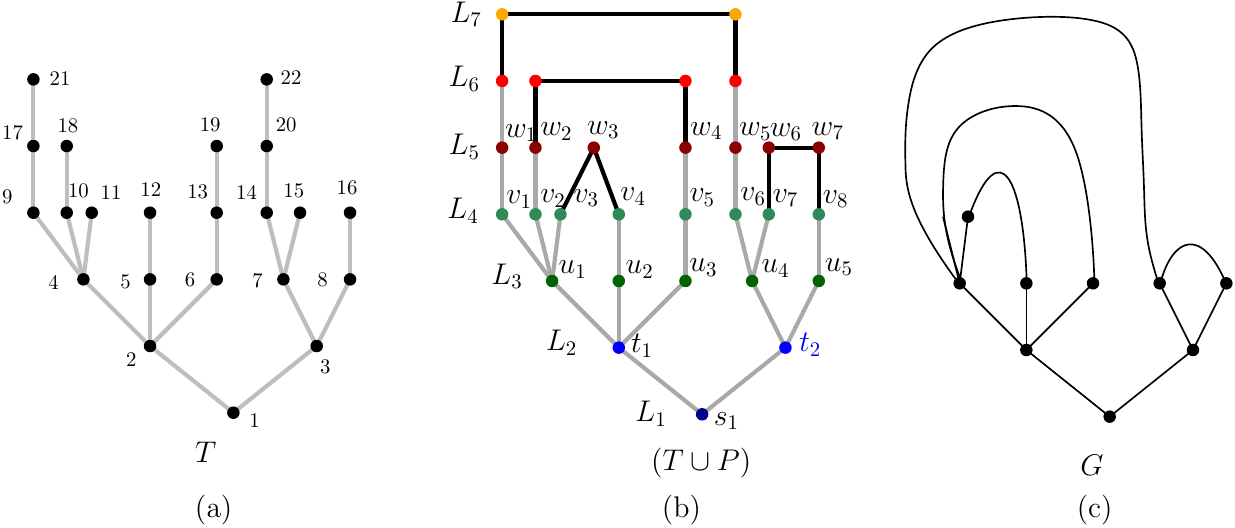}
    \caption{(a) An embedded tree. (b) An even subdivision of $ (T\cup P)$, where $T$ and $P$ are shown in gray and black, respectively. (c) The planar graph underlying $(T\cup P)$.}
    \label{fig:h0}
\end{figure}
\subsection{Even Subdivisions of Planar Graphs}  Let $T$ be an embedded tree and let $\sigma$ be the order of its nodes in breadth-first traversal that starts at the root and visits its children in clockwise order, and then at each internal node, considers its children in clockwise order starting from its  parent. Figure~\ref{fig:h0}(a) illustrates an embedded tree where nodes are ordered using such a breadth-first traversal. By \emph{the $i$th level} of $T$, where $i\ge 1$, we denote the ordered subset of nodes from $\sigma$ that are at distance $(i-1)$ from the root of $T$. By $\sigma_i$ we denote the ordered sequence of vertices at the $i$th level.    


An \emph{admissible extension} of $T$ is a set $P$ of vertex-disjoint paths with the following properties:   (a) Each maximal path in $P$ has 3 or 4 vertices. (b) The endpoints of each maximal path in $P$ are leaves of $T$ that are consecutive and at the same level. (c) The internal vertices of any path in $P$ are not vertices of $T$.

Figure~\ref{fig:h0}(b) shows an admissible extension of a tree $T$, where the edges of $T$ are shown in gray and the edges of $P$ are shown in black, respectively. By $\overline{(T\cup P)}$ we denote the complement of the graph $(T\cup P)$. An \emph{even subdivision} of a graph $G$ is a graph $G'$ that for each edge $(u,v)$ of $G$, either retains the edge or replaces it with a path of  even number of division vertices. We will use the following lemma.

\begin{lemma}[Cabello, Cardinal, and Langerman~\cite{DBLP:journals/dcg/CabelloCL13}] 
Any embedded planar graph $G$ has an even subdivision $(T \cup P)$, where $T$ is an embedded tree and $P$ is an admissible extension of $T$. Furthermore, such $T$ 
and $P$ can be computed in polynomial time.
\end{lemma}

\subsection{Computing an Upward Ray Representation of $\overline{(T\cup P)}$}

We now show that $\overline{(T\cup P)}$ (obtained by taking the complement of $(T\cup P)$) admits an intersection representation with upward rays. We begin by introducing some notation that will be used throughout the construction. Let $r$ be a ray with origin $q$. Let $h_q$ be a horizontal line through $q$. We say $r$ is \emph{upward}  if it has a non-zero slope and lies above   $h_q$. Let $A$ be an arrangement of a set of rays (Figure~\ref{fig:h01}(a)). By a \emph{face} or a \emph{cell} of $A$ we denote a connected region of the plane delimited by the rays of $A$. By a \emph{path of $A$} we denote a simple polygonal chain in $A$. A path is called an \emph{outerpath} if all its segments lie on the unbounded face of $A$.  
A \emph{sail configuration} $\mathcal{S}$ of size $k$ is an arrangement of $k$ mutually intersecting rays $r_1,\ldots,r_k$ that contains an outerpath $e_1,\ldots, e_k$, where for each $i$ from 1 to $k$, $e_i$ is a part of $r_i$ (Figure~\ref{fig:h01}(b)). We refer to $e_1,\ldots,e_k$ as the \emph{extremal path} of $\mathcal{S}$. A sail is called \emph{upward} if it consists only of upward rays. We will often add a sail $\mathcal{S}'$ on top of another sail $\mathcal{S}$ such that each ray $r'_i$ of $\mathcal{S}'$  originates at the cell immediately above the edge $e_i$ on the extremal path of $\mathcal{S}$ (Figure~\ref{fig:h01}(c)). Since $r'_i$ intersects all the other rays except for $r_i$,  this realizes the complement of a matching.

\begin{figure}[h]
    \centering
    \includegraphics[width=\linewidth]{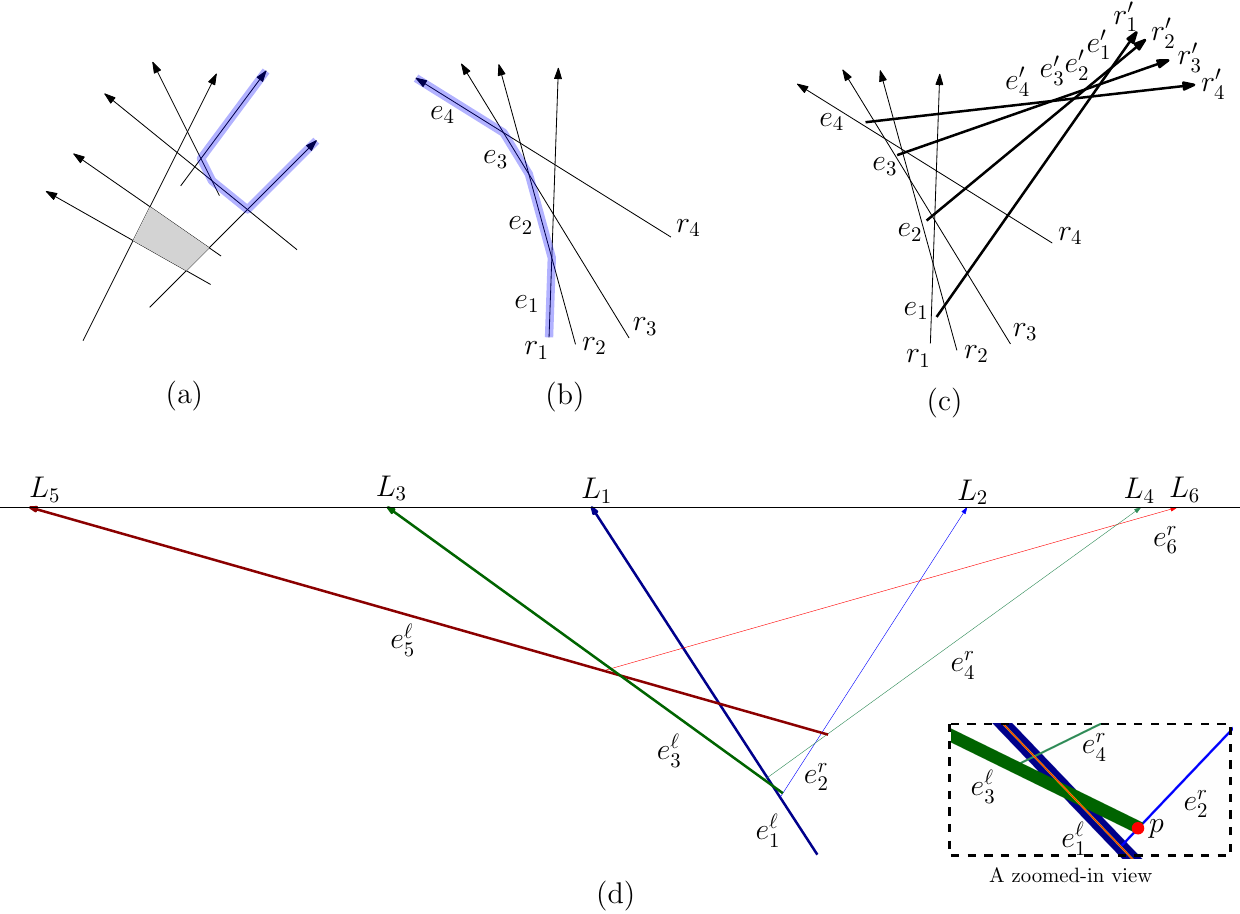}
    \caption{(a) Illustration for a face of $A$ in gray and an outerpath in blue shade. (b) An upward sail   $\mathcal{S}$, where the extremal path is in blue. (c) An upward ray representation of a complement of a matching using two sails. (d) An illustration for the sail skeleton.
    }
    \label{fig:h01}
\end{figure}

We are now ready to describe the details. We first refer the reader to  Figure~\ref{fig:h0}(b) which  illustrates an even subdivision $(T\cup P)$ of a planar graph, and Figure~\ref{fig:h01}(d) that shows a high-level sketch for the construction. We start with two upward sail configurations $\mathcal{S}_\ell$ and $\mathcal{S}_r$ of $\lceil k/2\rceil+1$ rays, where $k$ is the number of levels of $T$. 
We label the rays of  $\mathcal{S}_\ell = (L_1,L_3,\ldots)$ and $\mathcal{S}_r =(L_2,L_4,\ldots)$ with odd and even indices, respectively. 
Figure~\ref{fig:h0}(d) shows the first three rays of  $\mathcal{S}_\ell$ and  $\mathcal{S}_r$ in thick and thin rays, respectively.  
We refer to these two sails as the \emph{sail skeleton}, and their rays as \emph{skeleton rays}, which will guide the construction of the rays corresponding to $(T\cup P)$. For simplicity of presentation, we define more skeleton rays than the number of levels in $T$. Consequently, some of the high-index  skeleton rays may remain unused and can be removed after the complete construction. Assume that  $k'=\lceil k/2\rceil+1$. We define $L_1$ to be an upward ray starting at origin  with $(90^\circ+\alpha)$ angle of inclination, where  $\alpha = 90^\circ/2k'$. For each $i$ from $2$ to $2k'$, the origin of $L_i$ lies on $L_{i-1}$, and $L_i$ intersects all other skeleton  rays. We include the detailed construction of a sail skeleton in Appendix~\ref{app:sail}. 

\begin{figure}[h]
    \centering
    \includegraphics[width=\linewidth]{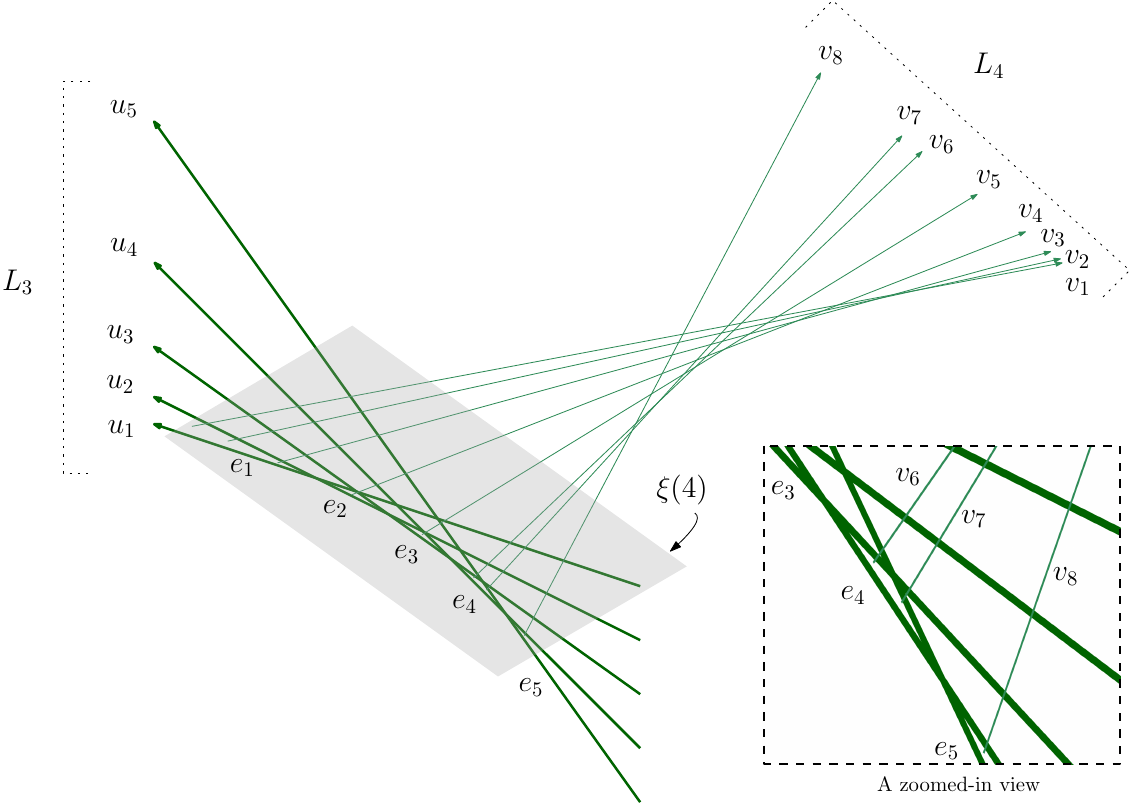}
    \caption{Construction of the vertices of two consecutive levels. }
    \label{fig:h1}
\end{figure}
\begin{figure}[h]
    \centering
    \includegraphics[width=\linewidth]{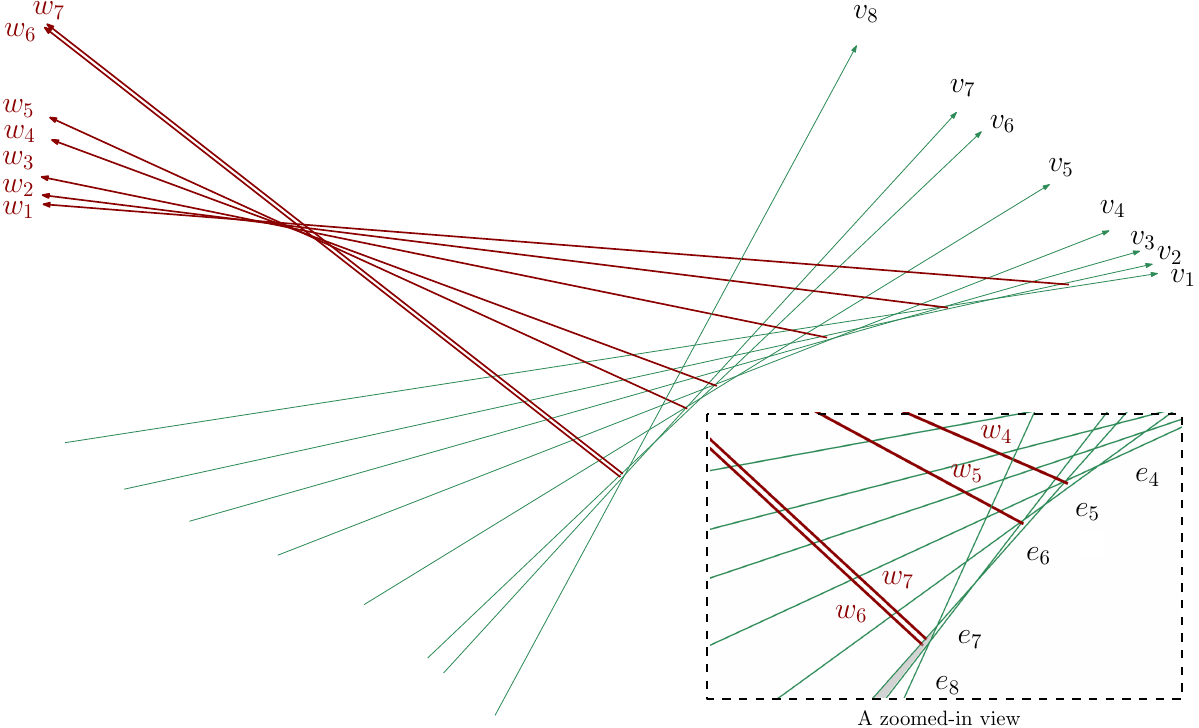}
    \caption{Illustration for the representation of some paths in $P$.}
    \label{fig:h2}
\end{figure}
\subsubsection{Construction of the Rays Corresponding to the Nodes of $T$} 

The skeleton ray $L_1$ represents the root of $T$. Let $\gamma$ be the smallest distance over all pairs of distinct intersection points in the sail skeleton. For each skeleton ray $L_i$, where $1\le i\le 2k'$,  we now consider a strip $L^s_i$ of width $\epsilon$, chosen sufficiently small relative to $\gamma$, so that the thickened rays do not completely cover any face of the sail skeleton.  

By $\xi(j)$, where $j\ge 2$, we denote the parallelogram determined by $L^s_{j-1}$ and $L^s_j$ (Figure~\ref{fig:h1}).  For each $j$th level of $T$, where $j\ge 2$, we now construct a sail configuration $R_j$ with $|\sigma_j|$ upward rays such that the origins of these rays remain bounded within $\xi(j)$ and all intersection points of the sail remain within $\xi(j+1)$. We will initially make all rays of $R_j$ to pass through a common point, and then move the rays to ensure a non-degenerate extremal path. 

Let $u_1,\ldots,u_r$ be the vertices in $\sigma_{j-1}$ and let $e_1,\ldots,e_r$ be the extremal path of sail $R_{j-1}$ (Figure~\ref{fig:h1}). For each $t$ from $1$ to $r$, we construct $c(u_t)$ upward rays, where $c(u_t)$ is the number of children of $u_t$. The origin of these rays lie inside the cell immediately above $e_t$ and the rays pass through the centroid of $\xi(j+1)$. For example, the vertex $u_4$ in Figure~\ref{fig:h1} has two children $v_6,v_7$, and thus the origins of their rays lie in the cell immediately above $e_4$.

We now consider a vertex $v$ in $\sigma_j$ and verify its adjacencies. Since we are realizing $\overline{(T\cup P)}$,   the ray for $v$ must intersect all rays corresponding to $\sigma_1,\ldots,\sigma_{j}$ except for the ray representing its parent $u_t$ in $T$. By the construction of the initial sail skeleton, the ray of $u_t$ intersects all rays  corresponding to the vertices in $\sigma_1,\ldots,\sigma_{j-2}$. Since the ray corresponding to $v$ starts at the cell immediately above $e_t$, it intersects all the rays corresponding to $\sigma_{j-1}\setminus u_t$. Finally, since the rays of $R_j$ pass through a common point of $\xi(j+1)$, the adjacencies for $v$ have been correctly realized. We now  move the rays of $R_j$  (without changing origins) to  obtain a non-degenerate sail configuration $R_j$ whose intersection points all lie within $\xi(j+1)$. Since we do not require  polynomially bounded coordinates, it is straightforward to move the rays one at a time in the same order as in $\sigma_j$ to obtain such a representation (Appendix~\ref{app:degenerate}).

\subsubsection{Construction of the Rays for the Internal Nodes of $P$}

The rays for the internal nodes of $P=(a,b,c,d)$ can be inserted in almost the same way as in the previous section. By definition of an admissible extension, the rays of $a,d$ appear consecutively on an extremal path, and hence the rays for $b,c$ can be constructed in the cell immediately above the edges $e_a$ and $e_d$ corresponding to $a,d$. Figure~\ref{fig:h2} illustrates an example where $P=(v_7,w_6,w_7,v_8)$. The details are included in Appendix~\ref{app:P}.

\smallskip
\noindent
The following theorem summarizes the result of this section. 

\begin{theorem}\label{thm:uprayhard}
    Finding a maximum clique is NP-hard for the intersection graphs of upward rays, grounded segments, and stabbed unit-length segments. 
\end{theorem}

\section{Approximation Results and Exact Computation}
In this section we give polynomial-time algorithms to compute a maximum clique in an intersection graph of grounded segments in two restricted cases. 
We then give a polynomial-time algorithm that achieves an $O(n^{3/4})$-approximation. 

\subsection{Segments with Endpoints on Fixed Line Set} 
\label{sec:fixedlineset}
The origin of the rays that we used to show NP-hardness for segment graphs 
cannot be covered with  $O(1)$ lines. Therefore, our  hardness result does not hold if we restrict the endpoints of the segments to lie on a fixed number of lines. We now show that such restrictions allow us to design a polynomial-time algorithm to  find a maximum clique in an intersection graph of grounded segments. However, the time-complexity for the stabbed setting under this restriction remains open. For simplicity of presentation, we assume that the endpoints of the segments are distinct. 

The idea is to design a dynamic programming algorithm by introducing a concept of strict independent set. We define a pair of grounded segments to be \emph{strictly independent} if they satisfy the following conditions:  ($C_1$) The segments  
do not intersect. ($C_2$) Extending either of the segments beyond the free endpoint (above the ground line) does not create an intersection with the other. However, if both segments are extended beyond their free endpoints, the extended parts intersect each other.  


The pairs of segments in Figures~\ref{fig:trans}(a)--(c) are not strictly independent, whereas Figure~\ref{fig:trans}(d) illustrates a strictly independent pair. 
 A \emph{strict independent set} is a collection of grounded segments where  every pair is strictly independent. We first transform the input grounded segments $S$ into another set of grounded segments $S'$ such that a strict independent set in $G(S')$  corresponds to a maximum clique in $G(S)$.    

\begin{figure}[h]
    \centering
    \includegraphics[width=.9\linewidth]{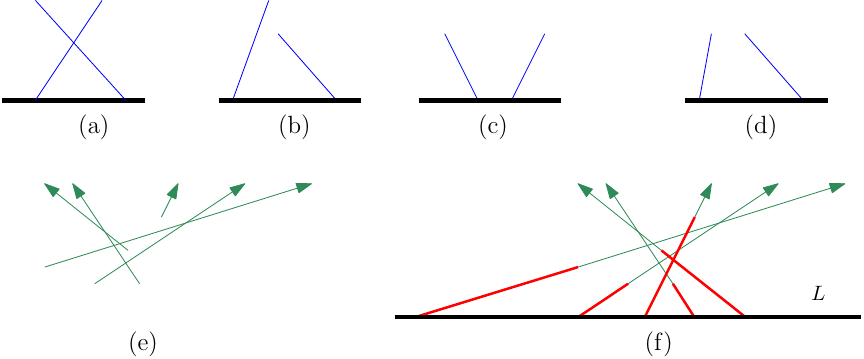}
    \caption{(a)--(c) Pairs 
    that are not strictly independent. (d) A strictly independent pair. 
    (e) An upward ray representation. 
    (f) A corresponding grounded segment representation 
    in red.  }
    \label{fig:trans}
\end{figure}
\subsubsection{Transforming $S$ into $S'$}
\label{sec:trans}
We use the result of~\cite{DBLP:journals/jgaa/CardinalFMTV18} to transform the segments of $S$ into a set of upward rays $U(S)$ such that two segments of $S$ intersect if and only if the corresponding rays of $U(S)$ intersect. Figure~\ref{fig:trans}(e) gives an example for $U(S)$. Such a transformation can be done in linear time.

We now consider the arrangement $\mathcal{A}$ of lines determined by the rays in $U(S)$ and define  a ground line $L$ such that all the intersection points of the arrangement lie above $L$. We now create the set $S'$ by considering the parts of the lines that are between the ground line $L$ and the rays in $U(S)$. Figure~\ref{fig:trans}(f) illustrates $S'$ in red. We now have the following lemma with the proof in Appendix~\ref{app:trans}. 


\begin{lemma}\label{lem:equiv}
A pair of segments in $S$ intersect if and only if the corresponding segments in $S'$ are strictly independent. 
\end{lemma}

\subsubsection{Dynamic Program to Compute a Strict Independent Set of $S'$}
\label{sec:dp}
Assume that the lines $\ell_1,\ldots ,\ell_k$ cover all the free endpoints of the segments in $S'$.
We define the \emph{height of a segment} $s$ to be the $y$-coordinate of its free endpoint and denote the height by $h(s)$. Let $\sigma[1\ldots n]$ be a  sequence of segments obtained based on their order on the ground line. We use the terms left and right to describe the relative order of segments in $\sigma$.  
The idea of the algorithm is to order the segments in non-decreasing order of heights. The algorithm will use this order to guess the segments in the optimal solution. Assume that a set $I$ of strictly independent segments has already been selected. To extend this set while maintaining strict independence, the algorithm generates subproblems by guessing the next tallest segment to include and takes the maximum over all such choices. Another key  component of the algorithm is a compact problem encoding. Instead of retaining all the segments of $I$, we will maintain a set of at most $4k$ segments to pass to each  subproblem.

We first introduce some notation and preliminary observations. 
Let $s_i,s_j$ be two segments where $s_i$ is to the left of $s_j$ in $\sigma$. Let $T$ be a strict independent set of at most $4k$ segments, where at most $2k$  lie in $\sigma[1\ldots i]$ and at most $2k$  lie in $\sigma[j\ldots n]$.  By  $Q[s_i ,s_j, T]$ we denote the maximum number of segments that can be added from $\sigma[i+1\ldots j-1]$, where no segment is higher than both $s_i,s_j$, and together they form a strict independent set with $T$ (Figure~\ref{fig:dp}(a)). 

\begin{figure}
    \centering
    \includegraphics[width=\linewidth]{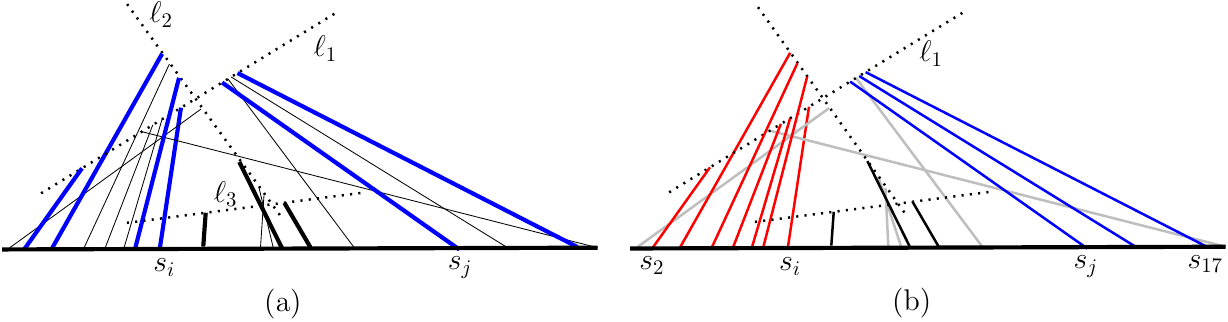}
    \caption{(a) A set of grounded segments with free endpoints covered by 3 lines (shown in dots). A set $T$ is shown in thick blue, and the segments corresponding to  $Q[s_i,s_j, T]$ are shown in thick black. (b) Illustration for $\overline{A'}_\ell,\overline{A'}_r$ and $A'$ in red, blue, and black, respectively.  }
    \label{fig:dp}
\end{figure}

Let $A$ be a maximum strict independent set in $S'$. Let $s_i$ and $s_j$ be two segments in $A$ such that no segment of $A$ between $s_i$ and $s_j$ is higher than both $s_i,s_j$. Let $A'$ be the set of these segments. Let $\overline{A'}_\ell$ and $\overline{A'}_r$ be the segments  of $(A\setminus A')$ that lie to the left and right of $A'$, respectively (Figure~\ref{fig:dp}(b)). 
We now construct a strict independent set $T$ as follows. For each $w$ from $1$ to $k$, we fill the entries  $T[4w-3\ldots 4w]$ from segments of $(A\setminus A')$ that have free endpoints on $\ell_w$. Specifically, the first two entries  $T[4w-3]$ and $T[4w-2]$ are set to the highest and lowest segments of $\overline{A'}_\ell$ that have free endpoints on $\ell_w$. Similarly, $T[4w-1]$ and $T[4w]$ are set to the highest and lowest segments of $\overline{A'}_r$ that have free endpoints on $\ell_w$. If we do not find enough segments to assign, then some of these entries remain null. It is straightforward to observe by the construction of $T$ that $s_i$ and $s_j$ belong to $T$. For example, in Figure~\ref{fig:dp}(b), the entries $T[1,\ldots,4]$ are set to $s_i,s_2,s_{17},s_j$, respectively.  

The following lemma  justifies the  problem  encoding, whose proof is   in Appendix~\ref{app:dp}.

\begin{lemma}\label{encoding}
Let $A$ be a maximum strict independent set in $S'$, and let $s_i$ and $s_j$ be two segments in $A$. Let $A'$ be the set of the segments of $A$ that lie between $s_i$ and $s_j$. Assume that no segment of $A'$ is higher than both $s_i,s_j$. Let $Z$ be the segments corresponding to $Q[s_i ,s_j, T]$. Then $Z\cup (A\setminus A')$ is a 
 maximum strict independent set of $S'$. 
\end{lemma}

\begin{figure}[h]
    \centering
    \includegraphics[width=\linewidth]{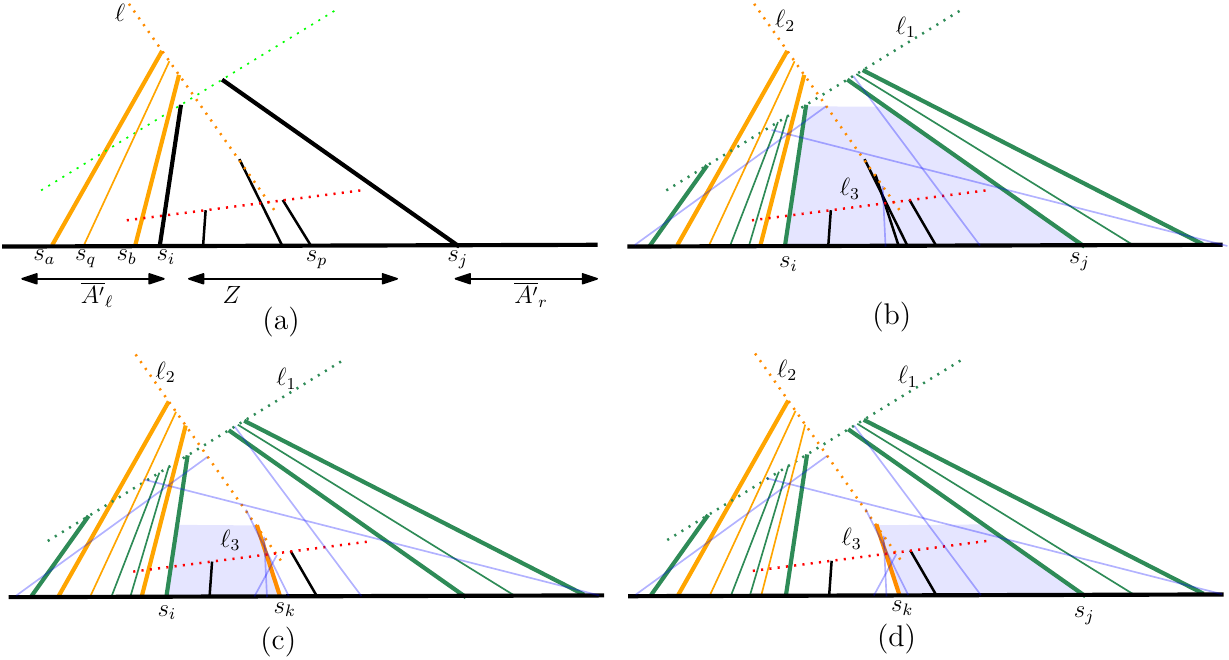}
\caption{(a) Illustration for Lemma~\ref{encoding}. 
 (b) A subproblem  $Q[s_i,s_j,T]$, where segments that have already been taken in the solution are shown in green or orange  depending on whether they are on $\ell_1$ or $\ell_2$, respectively. The segments that are not strictly independent with the selected segments are in thin blue. The segments in $T$ are shown in thick orange or green lines. The black segments are candidates for $s_k$. (c)--(d) A decomposition of $Q[s_i,s_j,T]$ into $Q[s_i,s_k,T_1]$ and $Q[s_k,s_j,T_2]$. }
    \label{fig:dp1}
\end{figure}

By Lemma~\ref{encoding}, instead of tracking the entire set  $(A\setminus A')$ that has already been taken in the solution, we can encode the subproblems  with respect to a subset $T$ of $(A\setminus A')$, i.e., $Q[s_i ,s_j, T]$. We now  formulate a recurrence relation for $Q[s_i ,s_j, T]$. For convenience of presentation, we insert two dummy segments in $S'$. One  to the left of $\sigma[1]$ and the other to the right of $\sigma[n]$ such that the following conditions hold for the newly added segments. (a) They are strictly independent both among themselves and with each segment in $S'$. (b) They are higher than the segments of $S'$. (c) Their free  endpoints lie on a horizontal line $\ell_{k+1}$ which is different than $\ell_1,\ldots,\ell_k$. 

The dummy segments $\sigma[0]$ and $\sigma[n+1]$ can be constructed by first computing a rectangle $abcd$ that encloses all the segments of $S'$ and intersects the extended parts of all the segments in $S'$ at its top side $ab$. The segments $\sigma[0]$ (resp., $\sigma[n+1]$) are then constructed with a positive (resp., negative) slope with $a$ (resp. $b$) as its free endpoint. 
The size of the maximum strict independent set in $S'$ is now $Q[s_0 ,s_{n+1}, T]$, where $T = \{s_0,s_{n+1}\}$. For simplicity, we write the elements of $T$ without showing the null entries.  The following lemma gives a  problem decomposition (Figures~\ref{fig:dp1}(b)--(d)) whose proof is   in Appendix~\ref{app:dp}.

\begin{lemma}\label{recurrence}
    \begin{align*}
    Q[s_i ,s_j, T] = 
    \begin{cases}
        0 & \text{; if no segment between $s_i,s_j$ has}\\&\text{height at most $\min \{h(s_i),h(s_j)\}$   } \\
        & \text{and is strictly  independent to $T$.}   \\
\max\limits_{s_k}\{Q[s_i ,s_k, T_1] + Q[s_k ,s_j, T_2] + 1\}&  \text{; over   $s_k$ between $s_i,s_j$ where}\\ & \text{$h(s_k){\le} \min \{h(s_i),h(s_j)\}$ and $s_k$ } \\
        & \text{is strictly independent with $T$.}   
    \end{cases}
\end{align*}
Let $\ell$ be the line that contains the free endpoint of $s_k$. Let $T^-$ ($T^+$) be the entries of $T$ that lie to the left (right) of $s_i$. If $T^-$ contains two segments $s,s'$ with free endpoints on $\ell$, then $T_1$ is obtained by adding $s_k$ and deleting the segment with median position among $s,s',s_k$. Otherwise, $T_1$ is obtained by adding $s_k$ to $T$. The construction of $T_2$ is symmetric with respect to $T^+$.  
\end{lemma}

Since there are $O(n)$ possibilities for each of $s_i$ and $s_j$, and $O(n^{4k})$ possibilities for $T$, the number of subproblems is  $O(n^{4k+2})$. By Lemma~\ref{recurrence}, each subproblem can be computed by $O(n)$ table lookup. If $k=O(1)$, the overall running time to compute the solution to all subproblems is $O(n^{O(1)})$. We now use Lemma~\ref{lem:equiv} to obtain the following theorem. 

\begin{theorem}\label{thm:fixedlines}
    Given a set of $n$ grounded segments whose free endpoints lie on $O(1)$ lines, a maximum clique in the corresponding intersection graph can be computed in $O(n^{O(1)})$ time. 
\end{theorem}


\subsection{Unit-Length Grounded Segments}\label{sec:unit}
We now give a polynomial-time algorithm to find a maximum clique in a grounded unit-length segment graph $G(S)$. Let $M$ be a maximum clique of $G(S)$, the idea is to define a set of subgraphs $\mathcal{G}$ of polynomial size such that at least one of them is guaranteed to contain the maximum clique $M$. The search for $M$ is then confined to these subgraphs. To ensure the feasibility of this approach, we also need to prove that each subgraph has a nice property that allows for finding its maximum clique in polynomial time. 

The construction of $\mathcal{G}$ is simple. For a pair $s_a$, $s_b$ of intersecting grounded segments, let $G(s_a,s_b)$ be the intersection graph determined by all the segments that intersect both $s_a$ and $s_b$, and lie between $s_a$ and $s_b$ on the ground line. We now define $\mathcal{G}$ to be the set of graphs $\bigcup_{s_a,s_b\in S} G(s_a,s_b)$. Let $M$ be a maximum clique of $G(S)$ and let $s_\ell$ and $s_r$ be the leftmost and rightmost segments of $M$, then the clique must exist in the graph $G(s_\ell,s_r)$. We will show that $G(s_\ell,s_r)$ is a cocomparability graph, and hence a maximum clique can be computed in polynomial time~\cite{gavril2000maximum}. 

A linear ordering $v_1\ldots,v_n$ of the vertices of a graph is called  \emph{an umbrella-free ordering} if for every $i,j,k$ with $1\le i < j < k\le n$, the existence of the edge  $(v_i,v_k)$ implies that  $v_j$ is  adjacent to at least one of $v_i$ or $v_k$. We will use the following  characterization of cocomparability graph. 

\begin{lemma}[Kratsch and Stewart~\cite{kratsch1993domination}]
\label{lem:umbrella}
A graph is a cocomparability graph if and only if it has an \emph{umbrella-free ordering}.   
\end{lemma}


\begin{lemma}\label{slide}
The graph $G(s_\ell,s_r)$ is a cocomparability graph.
\end{lemma}
\begin{proof}[Proof Sketch]
Let $\sigma=(s_\ell,\ldots,s_r)$ be the order of the segments of  $G(s_\ell ,s_r)$  on the ground line.  We   show that $\sigma$ is an umbrella-free ordering, and thus $G(s_\ell ,s_r)$ is a cocomparability graph by Lemma~\ref{lem:umbrella}. Specifically, for every subsequence $s_a,s_c,s_b$ in $\sigma$,  we   show that if $s_a$ and $s_b$ intersect, then at least one of them intersects $s_c$. We consider two cases based on the location of the intersection point $i(s_a,s_b)$ of $s_a$ and $s_b$. If $i(s_a,s_b)$ is above both the lines determined by $s_\ell$ and $s_r$ (Figures~\ref{fig:gr}(a)--(b)), then to avoid $s_a,s_b$, the length of $s_c$ must be less than one unit. If $i(s_a,s_b)$ is not above both lines (Figures~\ref{fig:gr}(c)--(d)), then to avoid  $s_a,s_b$, the segment $s_c$ must avoid at least one of $s_\ell$ and $s_r$. The details are in Appendix~\ref{app:cocom}. 
\end{proof}
\begin{figure}[h]
    \centering
    \includegraphics[width=.95\linewidth]{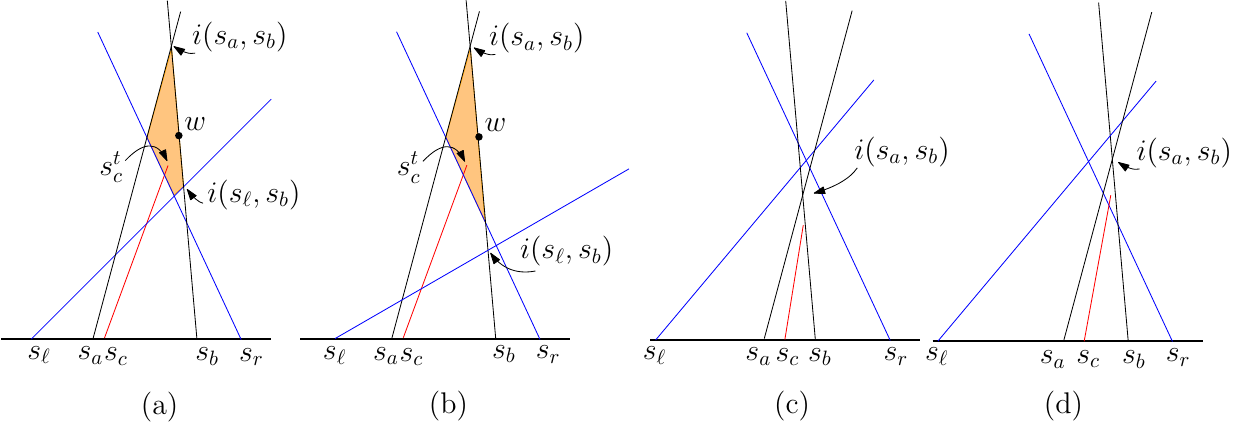}
    \caption{Illustration for the proof of Lemma~\ref{slide}. The  free endpoint $s^t_c$ of $s_c$ lies in orange region. 
    }
    \label{fig:gr}
\end{figure}

Since a maximum clique in an $n$-vertex cocomparability graph can be computed in polynomial time~\cite{gavril2000maximum}, we can use Lemma~\ref{slide} to obtain the following theorem.

\begin{theorem}\label{thm:unit}
Given a set of grounded unit-length segments, one can compute a  maximum clique in the corresponding intersection graph  in polynomial time.
\end{theorem}

\subsection{Approximation for Grounded Segment Graphs}\label{sec:approx}

We now give a polynomial-time $O(n^{3/4})$-approximation algorithm to compute a maximum clique in an  arbitrary grounded segment graph.

\begin{theorem}\label{thm:approx}
 Given $n$ grounded segments, one can compute a clique $C$ in their intersection graph in polynomial time such that $|C|$ is $\Omega(n^{-3/4}|M|)$, where $M$ is a maximum clique.
 \end{theorem}
\begin{proof}[Proof Sketch]  We partition the segments into $O(n^{3/4})$ subsets with specific property. Let $S$ be such a subset and let $\sigma$ denote an order of its segments  in increasing $x$-coordinates of the  free endpoints. We design  $S$ such that the corresponding sequence of free  (resp., grounded) endpoints is monotonically increasing or decreasing in $y$-coordinates (resp., $x$-coordinates). We show that $S$ is a cocomparability graph, which implies that a clique of size $\Omega(n^{-3/4}|M|)$ can be found by examining all subsets. The details are in Appendix~\ref{app:approx}.
\end{proof}


\section{Disk Intersection Graphs} 

In this section we present the results on disk intersection graphs.

\subsection{A Polynomial-Time Algorithm in Grounded Setting} \label{sec:gd}

\begin{theorem}\label{thm:polydisk}
    Given a set of grounded disks $\mathcal{D}$, a maximum clique in the corresponding disk intersection graph $G(\mathcal{D})$ can be computed in polynomial time.
\end{theorem}
\begin{proof}[Proof Sketch]
For each disk $D\in\mathcal{D}$, we compute a maximum clique $M(D)$ that contains  $D$ as its smallest disk. A maximum clique of $G(\mathcal{D})$ is then obtained by taking the largest among these cliques. 
To compute $M(D)$, we can only examine the set of disks  $S\subseteq \mathcal{D}$ that intersect $D$. We next show that the intersection graph of $S$ is   a cobipartite graph. The details are in Appendix~\ref{app:gd}. 
\end{proof}

\subsection{(3/2)-Approximation for $[1,3]$-Disk Graphs} 

In this section we consider \emph{$[1,3]$-disks}, which are   disks with radii in the interval $[1,3]$. Let $\mathcal{D}$ be a set of $[1,3]$-disks.    We now show how to compute a 3/2-approximation in the corresponding  disk intersection graph in $O(n^3f(n))$ time, where $f(n)$ is the time to compute a  maximum clique in a $n$-vertex cobipartite graph. We will use the following observation that given a set of disks pierced by three points, one can compute a 3/2-approximation for the maximum clique in $O(f(n))$ time. The idea is to first construct three graphs, each formed by discarding the disks that are uniquely hit one of the three points, and then take the largest clique among the three graphs. To obtain the final approximation algorithm, we will show the search for the maximum clique 
can be reduced to $O(n^3)$ subsets of mutually intersecting disks.

\subsubsection{Piercing Pairwise Intersecting $[1,3]$-Disks}
\label{sec:piercing}
\noindent
{\bf Technical Background.} We first introduce some notation and preliminary results. We call three pairwise intersecting disks $D_1,D_2,D_3$ a \emph{Helly triple} if they have a common intersection point. Otherwise, we call them a \emph{non-Helly triple}.  The  \emph{inscribed circle} of a non-Helly triple is the largest circle that can be contained within the region bounded by the three disks of the triple (Figure~\ref{fig:lemmas}(a)). Given a line segment $ab$, we use the notation $L(ab)$ to denote the line determined by $ab$. For an angle $\angle pqr$, we consider the clockwise angle from $qp$ to $qr$. By $W(\angle pqr)$, we denote a wedge-shaped region  that is swept by a rotating ray $qp$ clockwise until it aligns with $qr$.
We will use  the following four lemmas, where the proof of the last three is included in Appendix~\ref{app:piercing}.
\begin{figure}[h]
    \centering
    \includegraphics[width=\linewidth]{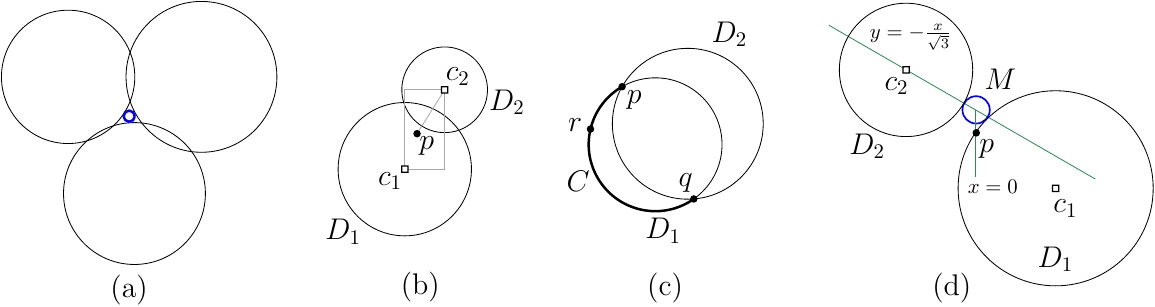}
    \caption{(a) A non-Helly triple and their inscribed circle are shown in black and blue, respectively. Illustration for (b) Lemma~\ref{lem:contain}, (c) Lemma~\ref{lem:arc}, and (d) Lemma~\ref{lem:120}. }
    \label{fig:lemmas}
\end{figure}
\begin{lemma}[Biniaz,  Bose and  Wang~\cite{biniaz2023simple}]\label{lem:hradi}
The inscribed circle of a non-Helly triple of three pairwise intersecting disks of radii $r $ has radius
at most $r(\frac{2}{\sqrt{3}}-1)$.
\end{lemma}

\begin{lemma}\label{lem:contain}
    Let $D_1$ and $D_2$ two intersecting disks with centers $c_1$ and $c_2$, respectively. Let $p$ be a point contained in the axis-aligned rectangle having $c_1c_2$ as its diagonal (Figure~\ref{fig:lemmas}(b)). If we move $c_2$ towards $p$ along $L(pc_2)$, then $D_2$ maintains the intersection with $D_1$.
\end{lemma}

\begin{lemma}\label{lem:arc}
    Let $D_1$ and $D_2$ be two disks that intersect at two distinct points $p$ and $q$. Let $C$ be an open circular arc on the boundary of $D_1$ between $p$ and $q$, and let $r$ be a point on $C$ (Figure~\ref{fig:lemmas}(c)). If $r$ lies outside of $D_2$, then the entire arc $C$ lies outside of $D_2$. 
\end{lemma}

\begin{lemma}\label{lem:120} 
Let $r$ be a number in the interval $[1,3]$. Let $M$ be a circle of radius $r(\frac{2}{\sqrt{3}}-1)$ with center $o=(0,0)$ and let $p$ be the point $(0,-\frac{1}{\sqrt{3}})$. Define two disks $D_1$ and $D_2$ as follows: $D_1$ touches $p$ and $M$ with its center lying in the fourth quadrant, and $D_2$ touches $M$ with its center lying on line $y=- \frac{x}{\sqrt{3}}$ in the second quadrant (Figure~\ref{fig:lemmas}(d)). If the radii of the disks are at most three, then $D_1$ cannot intersect $D_2$ below the line $y=- \frac{x}{\sqrt{3}}$. 
\end{lemma}

\noindent
{\bf Finding Piercing Points.} 
Let $Q$ be the smallest circle that intersects every disk in $\mathcal{D}$, which can be computed in linear time~\cite{loffler2010largest}. If $Q$ has zero radius, then a single point pierces all disks. Otherwise, $Q$ is an inscribed circle of a non-Helly triple of $\mathcal{D}$~\cite{biniaz2023simple}. The following lemma can now be used to find three piercing points $A,B,C$ for the case of unit disks.
 
\begin{lemma}[Biniaz,  Bose and  Wang~\cite{biniaz2023simple}]\label{lem:abc}
Let $S$ be a set of pairwise intersecting unit disks that contains a non-Helly triple. Assume that the smallest circle that intersects every disk in $S$ is centered at $(0,0)$. Then the points $A=(0,-\frac{1}{\sqrt{3}})$, $B=(1/2,1/2\sqrt{3})$ and $C=(-1/2,1/2\sqrt{3})$ pierce all the disks of $S$.
\end{lemma}

Assume that the smallest circle  that intersects all [1,3]-disks of $\mathcal{D}$ is $Q$, where $Q$ is determined by a non-Helly triple. Let $D_b, D_\ell, D_r$ be the disks of the non-Helly triple in clockwise order around $Q$  (Figure~\ref{fig:setup}(a)). Let  $c_b,c_\ell,c_r,c_q$  be the centers of $D_b,D_\ell,D_r,Q$, respectively. Let  $o$ be the origin  $(0,0)$ and assume that  $\angle c_\ell o c_r$ is the smallest angle among $\{\angle c_\ell o c_r,\angle c_r o c_b,\angle c_b o c_\ell\}$. 
We translate $\mathcal{D}$ such that the center of $Q$ coincides with  $o$, and rotate the plane such that  $c_b$ lies below $o$. The following property is known in this setting. 
\begin{lemma}[Carmi,  Katz  and Morin~\cite{carmi2023stabbing}]\label{lem:tan}
The tangent at the touching point of $D_l$ and $Q$ has a positive slope and the tangent at the touching point of   $D_r$ and $Q$ has a negative slope.
\end{lemma}

\begin{figure}
    \centering
    \includegraphics[width=\linewidth]{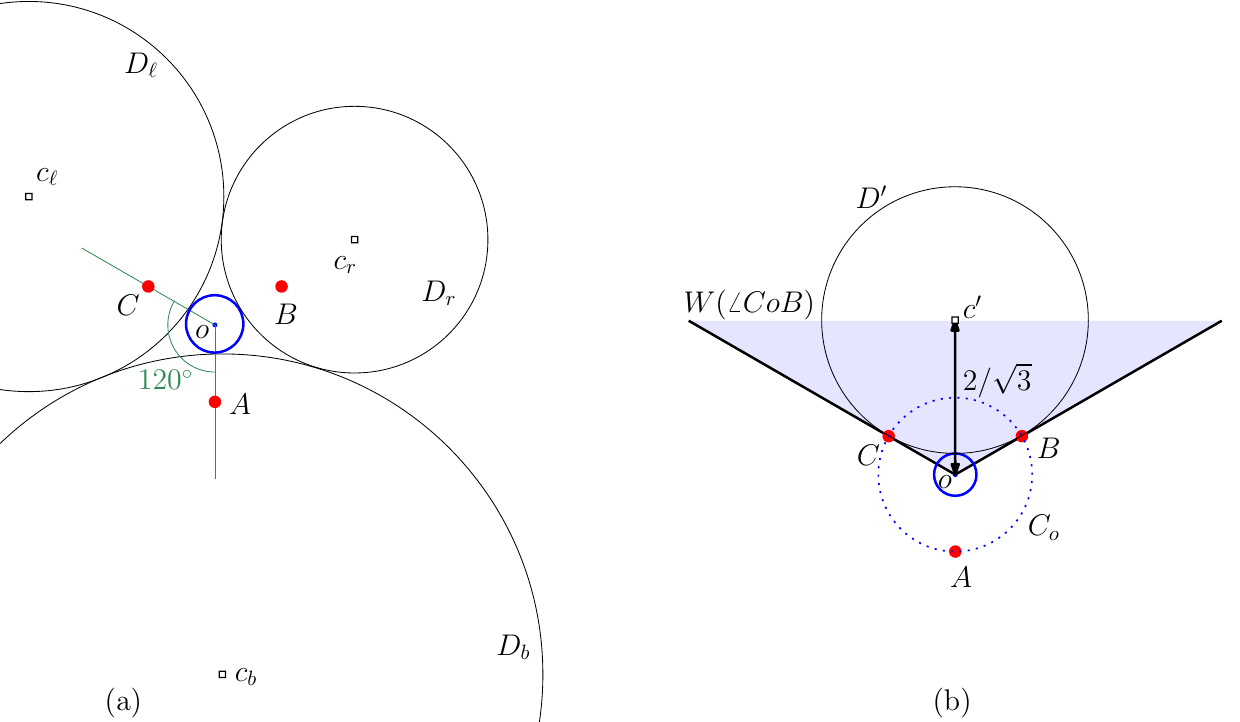}
    \caption{(a) Illustration for $\{D_b,D_\ell,D_r\}$, $\{A,B,C\}$ and $Q$, where $Q$ is in blue. Some disks are drawn partially for space constraints. (b) A unit disk touching  $B,C$ with center in $W(\angle CoB)$.}
    \label{fig:setup}
\end{figure}

We now show that the points $A,B,C$ of Lemma~\ref{lem:abc} pierce all disks in $\mathcal{D}$. 
Suppose for a contradiction that $D'$ is  a disk in $\mathcal{D}$ with center $c'$ and radius $r'$ such that $D'$ is not pierced by any point of $\{A,B,C\}$ but it intersects all three  disks $D_b,D_\ell,D_r$ and $Q$. Let $r_q$ be the radius of $Q$. Let $C_o$ be the circle of radius $\frac{1}{\sqrt{3}}$, i.e., a circle that  passes through $A,B,C$.
 
We first show that $r_q$ must be at least $ (\frac{2}{\sqrt{3}}-1)$ and $c'$ must lie outside of $C_o$.  The three bisectors of pairs of points from $\{A,B,C\}$ meet at $o$. Therefore, every point inside $C_o$ is within a distance of $\frac{1}{\sqrt{3}}$ from some point of $\{A,B,C\}$. Since $D'$ is not pierced, $c'$ cannot lie in $C_o$. Assume now that $c'$ lies outside of $C_o$. Since $D'$ is not pierced, we can move $D'$ such that $c'$ becomes equidistant from two of the points of $\{A,B,C\}$. Without loss of generality assume that $c'$ is equidistant from $B$ and $C$. If $D'$ is a unit disk, then    $|oc'|=\frac{2}{\sqrt{3}}$ (Figure~\ref{fig:setup}(b)). If $r_q$ is smaller than $ (\frac{2}{\sqrt{3}}-1)$, then $D'$ cannot intersect $Q$. If $D'$ has a larger radius,  then it is further away from $Q$,  contradicting that $D'$ intersects $Q$. 

Since the largest radius among the disks of $\mathcal{D}$ is three, the radius of $Q$ is at most $3(\frac{2}{\sqrt{3}}-1)$ (Lemma~\ref{lem:hradi}). The   remainder of the section assumes that $  (\frac{2}{\sqrt{3}}-1)\le r_q\le 3(\frac{2}{\sqrt{3}}-1)$ and   $c'$ is   outside of $C_o$. We now consider two cases based on the location of $c'$  and show that in each case  $D'$ is pierced.






\smallskip
\noindent
{\bf Case 1 ($c'\in W(\angle CoB)$).}   
Note that $D'$ cannot contain $o$. We now assume without loss of generality that  $D'$ intersects $D_b$ on the left halfplane of $L(oA)$, and hence below $L(oC)$. We translate $D'$ along the line $L(c'o)$ such that it touches $C$. Since $D'$ previously intersected $D_b$ and $Q$, by Lemma~\ref{lem:contain} it continues to intersect $D_b$ and $Q$ during this translation.  
Since $D'$ intersects $D_b$ and touches $C$, $c'$ is on the left halfplane of the line $x=C$. 
We now rotate $D'$ counter-clockwise without changing distance $|c'C|$ such that it touches $Q$ (Figure~\ref{fig:transf}(a)). Let $c'_{old}$ be the position of $c'$ before the rotation. Since $c'_{old}$ is on the left halfplane of the line $x=C$, $\angle c'c'_{old}c_b$ cannot be larger than $\angle c'_{old}c'c_b$. Hence $|c'c_b|\le |c_{old}c_b| $, and $D'$ continues to  intersect $D_b$.  We now reach a contradiction by  Lemma~\ref{lem:120} considering $C, Q, D'$ and $D_b$ as $p,M,D_2$ and $D_1$, respectively (Figure~\ref{fig:transf}(b)).

\begin{figure}[h]
    \centering
    \includegraphics[width=0.65\linewidth]{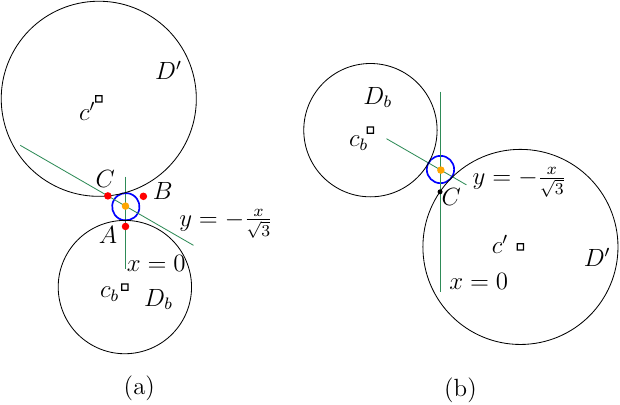}
    \caption{(a) Illustration for Case 1, where $Q$ is shown in blue and $o$ is marked in orange. (b) A transformation of the plane to show the relation to Lemma~\ref{lem:120}. }
    \label{fig:transf}
\end{figure}

\smallskip
\noindent
{\bf Case 2 ($c'\in W(\angle AoC)$ or $c'\in W(\angle BoA)$.} It suffices to consider the case when $c\in W(\angle BoA)$ because the argument in the other case is symmetric. 
Since $D'$ cannot contain $o$, we can consider two subcases depending on whether $D'$ intersects $D_\ell$ above or below $L(oC)$.

\smallskip
\noindent
{\bf Case 2.1 ($D'$ intersects $D_\ell$ below $L(oC)$).} We handle this case considering the following two scenarios. If $c_\ell$ lies above $L(oC)$, then we reposition $c_\ell$ and $D'$ so that $c_\ell$ lies on $L(oC)$ and $D'$ touches $A$ and   $Q$. We show how to reach a contradiction using Lemma~\ref{lem:120}.  If $c_\ell$ lies below  $L(oC)$, then we  observe by Lemma~\ref{lem:tan} that $D_\ell$ touches the upper-left arc of $D_b$. We then show that $D'$ avoids this arc, and hence avoids $D_\ell$. The details are   in Appendix~\ref{app:case2}. 
 
\smallskip 
\noindent
{\bf Case 2.2 ($D'$ intersects $D_\ell$ above $L(oC)$).}  The intuition for handling this case is as follows. Since $D_\ell$ touches $D_b$,  $\angle c_b o c_\ell$ cannot be very large. Furthermore, the constraint that $D'$ must intersect $Q$ and avoid $C$, pushes $D'$ away from the positive x-axis. We thus reach a contradiction by showing that $D'$ cannot intersect $D_\ell$. The details are   in Appendix~\ref{app:case2}.

\bigskip
\noindent
The following theorem summarizes the result of this section.
\begin{theorem}\label{thm:pier}
    Let $\mathcal{D}$ be a set of $n$ pairwise intersecting disks with radii in the interval $[1,3]$. Then one can compute three points that pierce $\mathcal{D}$ in $O(n)$ time.
\end{theorem}





\subsubsection{A 3/2-Approximation for the Maximum Clique in $G(\mathcal{D})$}

We begin by computing the arrangement of the disks and, for each cell, identify a clique determined by the set of disks intersecting that cell. We next enumerate all non-Helly triples, which requires   $O(n^3)$  guesses. For each such triple, we obtain three piercing points by Theorem~\ref{thm:pier}. For every pair of piercing points $p, p'$, we compute the maximum clique in the intersection graph induced by the disks pierced by $p$ and $p'$. Finally, we return the largest clique among all the cliques that have been computed. 

\begin{theorem}\label{thm:approxdisk}
Given  $n$ disks with radii in $[1,3]$, we can compute a clique $C$ in their intersection graph $G$  in $O(n^3f(n))$ time such that $|C|\ge 2|M|/3 $. Here $M$ is a maximum clique in $G$ and $f(n)$ is the time to find a maximum clique in an $n$-vertex cobipartite graph.
\end{theorem}
The theorem holds for any configuration of $[1,3]$-disks, and hence applies also in the scenario when all disks are stabbed by a line. 

\section{Conclusion}
There are several avenues for future research.  The case for stabbed segments with endpoints on $O(1)$ lines remains open. An intriguing direction is to establish APX-hardness or   design a PTAS for upward ray graphs. One may also attempt to generalize the results we obtained for $[1,3]$-disks. It would be interesting to find better approximations with faster running time.

\bibliographystyle{abbrv}
\bibliography{ref}

@article{gavril2000maximum,
  title={Maximum weight independent sets and cliques in intersection graphs of filaments},
  author={Gavril, Fanica},
  journal={Information Processing Letters},
  volume={73},
  number={5-6},
  pages={181--188},
  year={2000},
  publisher={Elsevier}
}

@article{erdos1935combinatorial,
  title={A combinatorial problem in geometry},
  author={Erd{\"o}s, Paul and Szekeres, George},
  journal={Compositio mathematica},
  volume={2},
  pages={463--470},
  year={1935}
}

@article{loffler2010largest,
  title={Largest bounding box, smallest diameter, and related problems on imprecise points},
  author={L{\"o}ffler, Maarten and van Kreveld, Marc},
  journal={Computational Geometry},
  volume={43},
  number={4},
  pages={419--433},
  year={2010},
  publisher={Elsevier}
}

@article{mondal2021simultaneous,
  title={Simultaneous Embedding of Colored Graphs},
  author={Mondal, Debajyoti},
  journal={Graphs and Combinatorics},
  volume={37},
  number={3},
  pages={747--760},
  year={2021},
  publisher={Springer}
}

@book{hadwiger1955ausgewahle,
  title={Ausgew{\"a}hle Einzelprobleme der kombinatorischen Geometrie in der Ebene},
  author={Hadwiger, Hugo and Debrunner, Hans},
  year={1955},
  publisher={Kundig}
}

@article{carmi2023stabbing,
  title={Stabbing pairwise intersecting disks by four points},
  author={Carmi, Paz and Katz, Matthew J and Morin, Pat},
  journal={Discrete \& Computational Geometry},
  volume={70},
  number={4},
  pages={1751--1784},
  year={2023},
  publisher={Springer}
}

@article{kratsch1993domination,
  title={Domination on cocomparability graphs},
  author={Kratsch, Dieter and Stewart, Lorna},
  journal={SIAM Journal on Discrete Mathematics},
  volume={6},
  number={3},
  pages={400--417},
  year={1993},
  publisher={SIAM}
}

@inproceedings{cabello2012clique,
  title={The clique problem in ray intersection graphs},
  author={Cabello, Sergio and Cardinal, Jean and Langerman, Stefan},
  booktitle={European Symposium on Algorithms (ESA)},
  pages={241--252},
  year={2012},
  organization={Springer}, 
}

@article{kratochvil1990independent,
  title={Independent set and clique problems in intersection-defined classes of graphs},
  author={Kratochv{\'\i}l, Jan and Ne{\v{s}}et{\v{r}}il, Jaroslav},
  journal={Commentationes Mathematicae Universitatis Carolinae},
  volume={31},
  number={1},
  pages={85--93},
  year={1990},
  publisher={Charles University in Prague, Faculty of Mathematics and Physics}
}

@article{DBLP:journals/comgeo/BoseCKM0MS22,
  author    = {Prosenjit Bose and
               Paz Carmi and
               J. Mark Keil and
               Anil Maheshwari and
               Saeed Mehrabi and
               Debajyoti Mondal and
               Michiel Smid},
  title     = {Computing maximum independent set on outerstring graphs and their
               relatives},
  journal   = {Comput. Geom.},
  volume    = {103},
  pages     = {101852},
  year      = {2022},  
}

@inproceedings{liu2023geometric,
  title={Geometric hitting set for line-constrained disks},
  author={Liu, Gang and Wang, Haitao},
  booktitle={Algorithms and Data Structures Symposium},
  pages={574--587},
  year={2023},
  organization={Springer}
}

@article{keil2017algorithm,
  title={An algorithm for the maximum weight independent set problem on outerstring graphs},
  author={Keil, J Mark and Mitchell, Joseph SB and Pradhan, Dinabandhu and Vatshelle, Martin},
  journal={Computational Geometry},
  volume={60},
  pages={19--25},
  year={2017},
  publisher={Elsevier}, 
}

@article{ambuhl2005clique,
  title={The clique problem in intersection graphs of ellipses and triangles},
  author={Amb{\"u}hl, Christoph and Wagner, Uli},
  journal={Theory of Computing Systems},
  volume={38},
  number={3},
  pages={279--292},
  year={2005},
  publisher={Springer}
}

@inproceedings{chan2018stabbing,
  title={Stabbing Rectangles by Line Segments--How Decomposition Reduces the Shallow-Cell Complexity},
  author={Chan, Timothy M and van Dijk, Thomas C and Fleszar, Krzysztof and Spoerhase, Joachim and Wolff, Alexander},
  booktitle={Proceedings of the 29th International Symposium on Algorithms and Computation (ISAAC)},
  year={2018}
}

@article{bandyapadhyay2019approximating,
  title={Approximating dominating set on intersection graphs of rectangles and L-frames},
  author={Bandyapadhyay, Sayan and Maheshwari, Anil and Mehrabi, Saeed and Suri, Subhash},
  journal={Computational Geometry},
  volume={82},
  pages={32--44},
  year={2019},
  publisher={Elsevier}
}

@article{DBLP:journals/jgaa/CardinalFMTV18,
  author    = {Jean Cardinal and
               Stefan Felsner and
               Tillmann Miltzow and
               Casey Tompkins and
               Birgit Vogtenhuber},
  title     = {Intersection Graphs of Rays and Grounded Segments},
  journal   = {J. Graph Algorithms Appl.},
  volume    = {22},
  number    = {2},
  pages     = {273--295},
  year      = {2018}, 
}

@inproceedings{FoxP11,
  author    = {Jacob Fox and
               J{\'{a}}nos Pach},
  editor    = {Dana Randall},
  title     = {Computing the Independence Number of Intersection Graphs},
  booktitle = {Proceedings of the Twenty-Second Annual {ACM-SIAM} Symposium on Discrete
               Algorithms (SODA)},
  pages     = {1161--1165},
  publisher = {{SIAM}},
  year      = {2011},  

}

@incollection{bang2006six,
  title={On six problems posed by Jarik Ne{\v{s}}et{\v{r}}il},
  author={Bang-Jensen, J{\o}rgen and Reed, Bruce and Schacht, Mathias and {\v{S}}{\'a}mal, Robert and Toft, Bjarne and Wagner, Uli},
  booktitle={Topics in Discrete Mathematics: Dedicated to Jarik Ne{\v{s}}et{\v{r}}il on the Occasion of his 60th Birthday},
  pages={613--627},
  year={2006},
  publisher={Springer}
}

@inproceedings{DBLP:conf/compgeom/KeilM25,
  author       = {J. Mark Keil and
                  Debajyoti Mondal},
  editor       = {Oswin Aichholzer and
                  Haitao Wang},
  title        = {The Maximum Clique Problem in a Disk Graph Made Easy},
  booktitle    = {Proceedings of the 41st International Symposium on Computational Geometry (SoCG)},
  series       = {LIPIcs},
  volume       = {332},
  pages        = {63:1--63:16},
  publisher    = {Schloss Dagstuhl - Leibniz-Zentrum f{\"{u}}r Informatik},
  year         = {2025}, 
}

@article{danzer1986losung,
  title={Zur l{\"o}sung des gallaischen problems {\"u}ber kreisscheiben in der euklidischen ebene},
  author={Danzer, Ludwig},
  journal={Studia Sci. Math. Hungar},
  volume={21},
  number={1-2},
  pages={111--134},
  year={1986}
}

@article{stacho1981solution,
  title={A solution of Gallai’s problem on pinning down circles},
  author={Stach{\'o}, Lajos},
  journal={Mat. Lapok},
  volume={32},
  number={1-3},
  pages={19--47},
  year={1981}
}

@article{biniaz2023simple,
  title={Simple linear time algorithms for piercing pairwise intersecting disks},
  author={Biniaz, Ahmad and Bose, Prosenjit and Wang, Yunkai},
  journal={Computational Geometry},
  volume={114},
  pages={102011},
  year={2023},
  publisher={Elsevier}
}

@inproceedings{DBLP:conf/cccg/cccg2024,
  author    = {Reymond Akpanya and Bastien Rivier and  Frederick B Stock},
  title     = {Open Problems from {CCCG} 2024}, 
  pages     = {167--170}, 
  booktitle     = {Proceedings of the 32nd Canadian Conference on Computational Geometry (CCCG) },
  year      = {2024}, 
}

@article{DBLP:journals/corr/abs-2107-05198,
  author    = {J. Mark Keil and
               Debajyoti Mondal and
               Ehsan Moradi and
               Yakov Nekrich},
  title     = {Finding a Maximum Clique in a Grounded 1-Bend String Graph},
  journal   = {Journal of Graph Algorithms and Applications},
  volume    = {26},
  number  = {4},
  year      = {2022},  
}

@article{c1,
  author    = {Sergio Cabello},
  title     = {Maximum clique for disks of two sizes},
  journal   = {Open problems from Geometric Intersection Graphs: Problems and Directions CG Week Workshop, Eindhoven, June 25}, 
  year      = {2015},  
}

@article{cabello2017refining,
  title={Refining the hierarchies of classes of geometric intersection graphs},
  author={Cabello, Sergio and Jej{\v{c}}i{\v{c}}, Miha},
  journal={The electronic journal of combinatorics},
  pages={P1--33},
  year={2017}
}

@article{c2,
  author    = {Sergio Cabello},
  title     = {Open problems presented at the Algorithmic Graph Theory},
  journal   = {Adriatic
 Coast workshop, Koper, Slovenia, June 16--19}, 
  year      = {2015},  
}

@article{DBLP:journals/dcg/CabelloCL13,
  author    = {Sergio Cabello and
               Jean Cardinal and
               Stefan Langerman},
  title     = {The Clique Problem in Ray Intersection Graphs},
  journal   = {Discret. Comput. Geom.},
  volume    = {50},
  number    = {3},
  pages     = {771--783},
  year      = {2013}, 
}

@article{MiddendorfP92,
  author    = {Matthias Middendorf and
               Frank Pfeiffer},
  title     = {The max clique problem in classes of string-graphs},
  journal   = {Discret. Math.},
  volume    = {108},
  number    = {1-3},
  pages     = {365--372},
  year      = {1992}, 

}

@article{agarwal2006independent,
  title={Independent set of intersection graphs of convex objects in {2D}},
  author={Agarwal, Pankaj K and Mustafa, Nabil H},
  journal={Computational Geometry},
  volume={34},
  number={2},
  pages={83--95},
  year={2006},
  publisher={Elsevier}
}

@article{chakraborty2024recognizing,
  title={Recognizing geometric intersection graphs stabbed by a line},
  author={Chakraborty, Dibyayan and Gajjar, Kshitij and Rusu, Irena},
  journal={Theoretical Computer Science},
  volume={995},
  pages={114488},
  year={2024},
  publisher={Elsevier}
}

@article{Jelinek019,
  author    = {V{\'{\i}}t Jel{\'{\i}}nek and
               Martin T{\"{o}}pfer},
  title     = {On Grounded {L}-Graphs and their Relatives},
  journal   = {Electr. J. Comb.},
  volume    = {26},
  number    = {3},
  pages     = {P3.17},
  year      = {2019}, 
}

@article{DBLP:journals/dm/ClarkCJ90,
  author    = {Brent N. Clark and
               Charles J. Colbourn and
               David S. Johnson},
  title     = {Unit disk graphs},
  journal   = {Discret. Math.},
  volume    = {86},
  number    = {1-3},
  pages     = {165--177},
  year      = {1990}, 
}

@phdthesis{breu, 
  author    = {Breu, Heinz},
  title     = {Algorithmic aspects of constrained unit disk graphs},
  school   = {University of British Columbia},
  year      = {1996},   
}

@article{DBLP:journals/dcg/EppsteinE94,
  author       = {David Eppstein and
                  Jeff Erickson},
  title        = {Iterated Nearest Neighbors and Finding Minimal Polytopes},
  journal      = {Discret. Comput. Geom.},
  volume       = {11},
  pages        = {321--350},
  year         = {1994}, 
}

@inproceedings{DBLP:conf/waoa/Fishkin03,
  author       = {Aleksei V. Fishkin},
  editor       = {Klaus Jansen and
                  Roberto Solis{-}Oba},
  title        = {Disk Graphs: {A} Short Survey},
  booktitle    = {Proceedings of Approximation and Online Algorithms, First International Workshop (WAOA)},
  series       = {Lecture Notes in Computer Science},
  volume       = {2909},
  pages        = {260--264},
  publisher    = {Springer},
  year         = {2003}, 
}

@inproceedings{DBLP:conf/compgeom/EspenantKM23,
  author       = {Jared Espenant and
                  J. Mark Keil and
                  Debajyoti Mondal},
  editor       = {Erin W. Chambers and
                  Joachim Gudmundsson},
  title        = {Finding a Maximum Clique in a Disk Graph},
  booktitle    = {Proceedings of the 39th International Symposium on Computational Geometry (SoCG)},
  series       = {LIPIcs},
  volume       = {258},
  pages        = {30:1--30:17},
  publisher    = {Schloss Dagstuhl - Leibniz-Zentrum f{\"{u}}r Informatik},
  year         = {2023}, 
}

@article{DBLP:journals/jacm/BonamyBBCGKRST21,
  author    = {Marthe Bonamy and
               {\'{E}}douard Bonnet and
               Nicolas Bousquet and
               Pierre Charbit and
               Panos Giannopoulos and
               Eun Jung Kim and
               Pawel Rzazewski and
               Florian Sikora and
               St{\'{e}}phan Thomass{\'{e}}},
  title     = {{EPTAS} and Subexponential Algorithm for Maximum Clique on Disk and
               Unit Ball Graphs},
  journal   = {J. {ACM}},
  volume    = {68},
  number    = {2},
  pages     = {9:1--9:38},
  year      = {2021},  
}

@inproceedings{DBLP:conf/focs/BonamyBBCT18,
  author       = {Marthe Bonamy and
                  Edouard Bonnet and
                  Nicolas Bousquet and
                  Pierre Charbit and
                  St{\'{e}}phan Thomass{\'{e}}},
  editor       = {Mikkel Thorup},
  title        = {{EPTAS} for Max Clique on Disks and Unit Balls},
  booktitle    = {Proceedings of the 59th {IEEE} Annual Symposium on Foundations of Computer Science (FOCS)},
  pages        = {568--579},
  publisher    = {{IEEE} Computer Society},
  year         = {2018}, 
}

@article{DBLP:journals/dam/FelsnerMW97,
  author    = {Stefan Felsner and
               Rudolf M{\"{u}}ller and
               Lorenz Wernisch},
  title     = {Trapezoid Graphs and Generalizations, Geometry and Algorithms},
  journal   = {Discret. Appl. Math.},
  volume    = {74},
  number    = {1},
  pages     = {13--32},
  year      = {1997},  
}

@inproceedings{DBLP:conf/compgeom/BonnetG0RS18,
  author    = {{\'{E}}douard Bonnet and
               Panos Giannopoulos and
               Eun Jung Kim and
               Pawel Rzazewski and
               Florian Sikora},
  editor    = {Bettina Speckmann and
               Csaba D. T{\'{o}}th},
  title     = {{QPTAS} and Subexponential Algorithm for Maximum Clique on Disk Graphs},
  booktitle = {Proceedings of the 34th International Symposium on Computational Geometry (SoCG)},
  series    = {LIPIcs},
  volume    = {99},
  pages     = {12:1--12:15},
  publisher = {Schloss Dagstuhl - Leibniz-Zentrum f{\"{u}}r Informatik},
  year      = {2018},  
}

\newpage
\appendix

\section{Construction of a Sail Skeleton}\label{app:sail}

Assume that  $k'=\lceil k/2\rceil+1$. We will ensure the sail-skeleton satisfies the following properties. 
\begin{enumerate}
    \item[P1.] For each $i$ from $1$ to $k'$, the absolute values of the slopes of $L_{2i-1}$ and $L_{2i}$ are the same.
    \item[P2.] For each $i$ from $2$ to $k'$, the slope of $L_{2i}$ is larger than $L_{2(i-1)}$.
    \item[P3.]  For each $i$ from $2$ to $2k'$, the origin of $L_i$ lies on $L_{i-1}$, and $L_i$ intersects all other skeleton  rays, i.e.,  $\{L_1,L_2,\ldots,L_{2k'}\}\setminus L_i$.
\end{enumerate}

We now describe an incremental construction for $\mathcal{S}_\ell$ and $\mathcal{S}_r$ that satisfies the properties P1--P3. Assume that  $\alpha = 90^\circ/2k'$ and let $\beta$ be a small positive constant.  

Let $L_1$ be the upward skeleton ray with origin at   $(0,0)$ and $(90^\circ+\alpha)$ angle of inclination with the positive x-axis. Let $z$ be a point on $L_1$. We define  $L_2$ to be the skeleton ray with origin $z$ and angle of inclination $(90^\circ-\alpha)$. It is straightforward to verify properties P1--P3 for $\{L_1, L_2\}$. For each $i$ from $2$ to $k'$, we set the angle of inclination of $L_{2i-1}$ and $L_{2i}$ to  $(90^\circ+i\alpha)$ and $(90^\circ-i\alpha)$, respectively, which satisfies P1--P2. We now describe the origins of $L_{2i-1}$ and $L_{2i}$ and show that $\{L_1, L_2, \ldots, L_{2i}\}$ satisfy P3. 

Let $e^\ell_1,e^\ell_3,\ldots,e^\ell_{2i-1}$ be the  extremal path of $\mathcal{S}_\ell$, and let $e^r_2,e^r_4,\ldots,e^r_{2i-2}$ be the  extremal path of $\mathcal{S}_r$. If $i$ is even, then let $p$ be a point on $e^r_{2i-2}$ located at  distance $\beta$ from the lower endpoint of $e^r_{2i-2}$ (e.g., the zoomed-in view of Figure~\ref{fig:h01}(d)). We set the origin of $L_{2i-1}$ as $p$. Since $L_{2i-1}$ has a negative slope and $L_2,\ldots,L_{2i-2}$ have positive slopes, $L_{2i-1}$ intersects all these skeleton  rays with even indices. Since $L_{2i-1}$ has a higher angle of inclination, it intersects all skeleton rays  $L_1,\ldots,L_{2i-3}$. Similarly, when $i$ is odd, we set  the origin of $L_{2i}$ at a point on $e^r_{2i-1}$ located at  distance $\beta$ from the lower endpoint of $e^r_{2i-1}$ and show that it intersects all skeleton rays in $\{L_1,\ldots,L_{2i-1}\}$. Consequently, property P3 becomes satisfied.

\section{Construction of a Non-degenerate Sail}
\label{app:degenerate}
Let $v_1,v_2,\ldots,v_|\sigma_j|$ be the vertices of $\sigma_j$ and let $r_1,r_2,\ldots,v_|\sigma_j|$ be the corresponding rays. Initially all the rays pass through the centroid of $\xi(j+1)$ (Figure~\ref{fig:h11}).  We do not move the ray for $v_1$ and $v_|\sigma_j|$. For $i$ from 2 to $|\sigma_j|-1$, we move $r_i$ closer to $r_1$ to form the edge $e'_i$ by intersecting $e'_{i-1}$. Since we do not require polynomially bounded coordinates, this process yields the required non-degenerate sail. 
 
\begin{figure}[h]
    \centering
    \includegraphics[width=\linewidth]{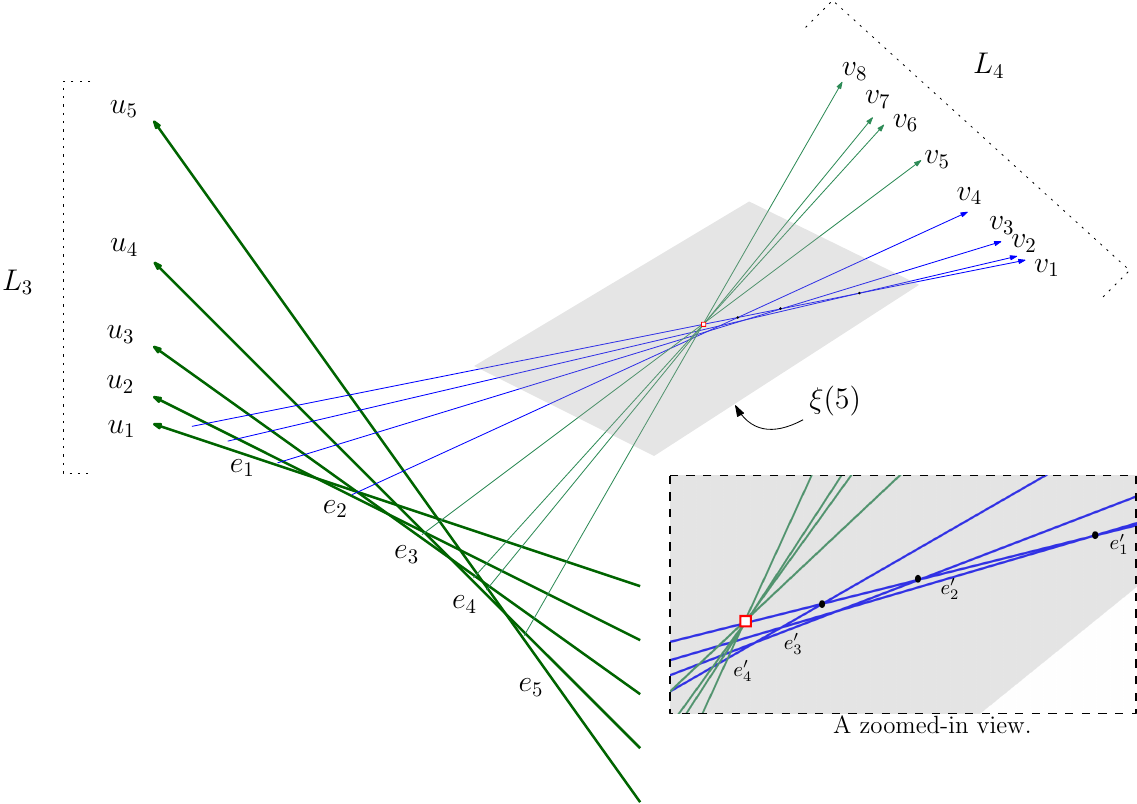}
    \caption{Illustration for the construction of a non-degenerate sail. Initially all the rays pass though the centroid of $\xi(5)$, which is marked in red square. The ray $v_1$ and $v_8$ stay at the same place. The rays $v_2,v_3,\ldots$ move  closer to intersect $v_1$ above the centroid forming the extremal path. The three rays $v_2,v_3,v_4$ moved so far are shown in blue. The intersection points between $v_1$ and $v_2,v_3,v_4$ are shown in black dots.}
    \label{fig:h11}
\end{figure}

\section{Construction of the Rays for the Internal Nodes of $P$}\label{app:P}

By the definition of admissible extension, each path in $P$ has either one or two internal vertices. Furthermore, its end vertices are two leaves of $T$ that are consecutive at  the same level of $T$. Let $a,b,c,d$ be a path in $P$, where $b$ and $c$ may coincide. Assume that $a,d$ belong to the $j$th level of $T$. Let $r_a,r_d$ be the rays of $a$ and $d$, respectively. We insert the rays $r_b$ and $r_c$ (for $b$ and $c$)  when constructing $R_{j+1}$. Let $e_a$ and $e_d$ be the segments corresponding to $r_a,r_d$ on  the extremal path of $R_j$. Since $a,d$ are consecutive in $\sigma_j$, the edges $e_a,e_d$  are consecutive on the extremal path of $R_j$.

We set the origins of both $r_b,r_c$  to a common point inside the cell immediately above $e_a$ and $e_d$. Similar to the construction in the previous section, the rays initially pass through the centroid  of $\xi(j+2)$. If $b\not=c$, then after constructing  $R_{j+1}$, we shift $r_c$ slightly by keeping it parallel to  $r_b$. We ensure that the  origin of $r_c$ remains within the same cell and it intersects the identical set of rays as $r_b$ (in the same order).  For example, consider  the path $(v_7,w_6,w_7,v_8)\in P$ in  Figure~\ref{fig:h2}, where the two rays of $w_6,w_7$ are parallel and originate inside  the cell  immediately above $e_{v_7},e_{v_8}$.

It now suffices to show that the adjacencies of $r_b$ and $r_c$ in $\overline{(T\cup P)}$ have been realized correctly, i.e., they   intersect  all rays corresponding to $\sigma_1,\ldots,\sigma_{j}$ except for the rays $r_a,r_d$. By the construction of the initial sail skeleton, $r_b,r_c$ intersect all rays in $\sigma_1,\ldots,\sigma_{j-2}$. Since $r_b,r_c$ start at the cell immediately above $e_a$ and $e_d$ (e.g., the cell shaded in gray in the zoomed-in view of Figure~\ref{fig:h2}), they intersect all  rays corresponding to $\sigma_{j-1}\setminus \{r_b,r_c\}$. Since $r_b,r_c$ initially pass through a common point of $\xi(j+2)$, they intersect all other rays in $\sigma_{j+1}$. The intersection graph remains the same during the perturbation that yields  $R_{j+1}$. Finally, in the case when $b\not= c$, we made $r_b$  and $r_c$ parallel. Consequently, all adjacencies of $b$ and $c$ in $\overline{(T\cup P)}$ are correctly realized.  

\section{Details of Section~\ref{sec:trans}}
\label{app:trans}
\begin{proof}[Proof of Lemma~\ref{lem:equiv}]
Let $r_1$ and $r_2$ be two rays in $U(S)$ and let $s_1$ and $s_2$ be the corresponding segments  in $S'$. 
If two segments in $S$ intersect, then their rays $r_1$ and $r_2$ in $U(S)$ intersect. By construction, $s_1$ and $s_2$ satisfy $(C_1)$--$(C_2)$, i.e., they are strictly independent. Consider now the  case when $s_1$ and $s_2$ are strictly independent. By conditions  $(C_1)$--$(C_2)$, extending the segments $s_1$ and $s_2$ beyond their free endpoints must create an intersection between the extended parts. Hence  $r_1$ and $r_2$ must intersect and consequently, their associated segments in $S$ intersect. 
\end{proof}
  
\section{Details of Section~\ref{sec:dp}}
\label{app:dp}

\begin{proof}[Proof of Lemma~\ref{encoding}]
We first show that $Z\cup (A\setminus A')$ is a strict independent set, i.e., every segment $s_p\in Z$ is strictly independent with every segment  $s_q\in (A\setminus A')$.
Since $(T\cup Z)$ is a strict independent set, we may assume that $s_q\not\in T$. Without loss of generality let $s_q$ be a segment of $\overline{A'}_\ell$ with the free endpoint on some line $\ell \in \{\ell_1,\ldots,\ell_k\}$ (Figure~\ref{fig:dp1}(a)). Since $s_q\not\in T$, $T$ must contain two segments $s_a$ and $s_b$, which are the highest and lowest segments of $\overline{A'}_\ell$  with free endpoints on $\ell$. Suppose for a contradiction that $s_p$ and $s_q$ are not strictly independent. 

Assume first that condition $C_1$ is violated, i.e., $s_p$ intersects $s_q$. Since the segments $s_a, s_q, s_b, s_i$ belong to $A$, they form a strict independent set. The segment $s_i$ lies between $s_q$ and $s_p$. Hence to avoid intersection with $s_q$ and $s_p$, $s_i$ must be smaller than $s_p$, which contradicts our assumption that  $s_p \in Z$ because no segment of $Z$ is higher than $s_i,s_j$. 

Assume now that condition $C_2$ is violated.  If the extended part of $s_q$ intersects $s_p$, then $s_p$ must be higher than $s_i$, which contradicts  our assumption that  $s_p \in Z$. 

If the extension of $s_p$ intersects $s_q$, then we show that $s_p$ cannot be strictly independent with $\{s_a,s_b\}$, which contradicts that $(T\cup Z)$ is a strict independent set.  Since $s_a$ and $s_b$ are strictly independent, they cannot lie on opposite halfplanes of $\ell$, as this would violate condition $C_2$. Consider now without loss of generality that $s_a$ and $s_b$ lie to the left halfplane of $\ell$. We now consider the following two cases.

\begin{enumerate}
    \item[-] Assume that $s_a$ is to the left of $s_q$. Since  $s_a$ is not intersected by the extension of $s_p$,   the free endpoint of $s_a$ lies on the left halfplane of the line $L_p$ determined by $s_p$, whereas the free endpoints of $s_q$ and $s_b$ lie on the right halfplane of $L_p$. Consequently, $s_b$ must intersect $s_p$ or its extension.

    \item[-] If $s_a$ is between $s_q$ and $s_p$, then to avoid intersection with the extension of $s_p$, $s_a$ must be smaller than $s_q$. This contradicts our  assumption that $s_q$ is smaller than $s_a$. 
\end{enumerate}

\noindent
We have now proved that $Z\cup (A\setminus A')$ is a strict independent set. Since $Q[s_i ,s_j, T]$ maximizes the number of segments that can be added to $T$ such that the resulting set is strictly independent, $|Z|$ cannot be smaller than $|A'|$. If $|Z|>|A'|$, then it contradicts the assumption that $A$ is a maximum strict independent set. Hence $Q[s_i ,s_j, T]$ is equal to $|A'|$.
\end{proof}

\begin{proof}[Proof of Lemma~\ref{recurrence}]
By definition,  $Q[s_i ,s_j, T]$  denotes the maximum number of segments from $\sigma[i+1\ldots j-1]$  that are no higher than $s_i,s_j$ and can be added to $T$ such that the resulting set is strictly independent. Therefore, in the base case, where no such segment exists, we set $Q[s_i ,s_j, T]$ to 0.

We now consider the general step (Figures~\ref{fig:dp1}(b)--(d)).  Let $Z$ be the set of segments corresponding to $Q[s_i ,s_j, T]$, which by definition forms strictly independent set with $T$. Let $s_k$ be the highest segment of $Z$ that is no higher than $s_i,s_j$. 
Let $Z_\ell$ and $Z_r$ be the segments of $Z$ that lie to the left and right of $s_k$, respectively. It now suffices to show  that $Q[s_i,s_k,T_1] = |Z_\ell|$ and the segments corresponding to $Q[s_i,s_k,T_1]$ are strictly independent with  $Z_r$. An analogous argument can be used to show $Q[s_i,s_k,T_2] = |Z_r|$. Together with $s_k$, we get the desired recurrence relation. 



Since $s_k\in Z$, $(Z_\ell \cup\{s_k\})$ is a strict independent set. Since $(T_1\setminus \{s_k\})\subset T$, the set $(Z_\ell \cup (T_1\setminus \{s_k\}))$ is strictly independent. Therefore, $(Z_\ell \cup T_1)$ is a strict independent set. Since $s_k$ is a highest segment in $Z$, no segment in $Z_\ell$ is higher than $s_k$. Observe that the segments corresponding to $Q[s_i,s_k,T_1]$ and $Z_r$ lie to the left and right of $s_k$, respectively, and all these segments are strictly independent to  $s_k$. Furthermore, none of these segments is higher than $s_k$. Therefore, the segments corresponding to $Q[s_i,s_k,T_1]$ are strictly independent with the set $Z_r$. By definition, $Q[s_i,s_k,T_1]$ maximizes the number of segments that satisfy the same hight and strict independence properties as that of $Z_\ell$. Therefore, $Q[s_i,s_k,T_1]\ge |Z_\ell|$. If $Q[s_i,s_k,T_1] > |Z_\ell|$, then we could delete $Z_\ell$ from $Z$ and add  all the segments corresponding to $Q[s_i,s_k,T_1]$ to $Z$ to improve $Q[s_i,s_j,T]$, which contradicts our initial assumption about the optimality of $Q[s_i,s_j,T]$. 
\end{proof}

\section{Details of Section~\ref{sec:approx}}
\label{app:approx}
\begin{proof}[Proof of Theorem~\ref{thm:approx}]  
Let $S$ be a set of $n$ grounded segments and let $G(S)$ be the corresponding intersection graph. Assume that the free endpoints of the segments are in general position, i.e., no two share the same $x$- or $y$-coordinates and no segment contains an endpoint of another segment. Let $F$ be the sequence of free endpoints in increasing order of $x$-coordinates. By Erd\H{o}s-Szekeres theorem~\cite{erdos1935combinatorial}, there exists a subsequence of $\Omega(\sqrt{n})$ free endpoints in $F$ that have either monotonically increasing or decreasing $y$-coordinates. One can use this result to obtain a partition of $F$ into subsequences $F_1,\ldots,F_k$, where $k=O(\sqrt{n})$, and every $F_i$, where $1\le i\le k$, contains at most $\sqrt{n}$ elements~\cite{mondal2021simultaneous}. Similarly, we partition the segments of each $F_i$ into a set of subsequences $\sigma_{i,1},\ldots, \sigma_{i,z}$, where $z=O(\sqrt[4]{n})$, and for each $\sigma_{i,j}$, where $1\le j\le k$, the $x$-coordinates of its segments on the ground line are either increasing or decreasing. 

Let $M$ be a maximum clique in $G(S)$. Then there exist  indices $i,j$, such that $\sigma_{i,j}$ contains at least $\frac{|M|}{zk} = \Omega(n^{-3/4}|M|)$ vertices of $M$. Let $M(\sigma_{i,j})$ be the maximum clique in the intersection graph corresponding to the segments of $\sigma_{i,j}$. Then $|M(\sigma_{i,j})|\ge \frac{|M|}{zk}$. We now show that $M(\sigma_{i,j})$ is cocomparability graph and hence it suffices to return a clique found over all subsequences. 

By Lemma~\ref{lem:umbrella},  it suffices to show that $\sigma_{i,j}$ is an umbrella-free ordering for the intersection graph corresponding to $\sigma_{i,j}$, and hence $M(\sigma_{i,j})$ can be computed in polynomial time~\cite{gavril2000maximum}.  Let $s_p,s_q,s_r$ be three segments in this order in   $\sigma_{i,j}$. Assume for a contradiction that $s_p$ and $s_r$ intersect, but none of them intersects $s_q$. Since $\sigma_{i,j}$ is a subsequence of $F_i$, the free endpoints $s^t_p,s^t_q,s^t_r$ are monotonically increasing or decreasing. Therefore, the $y$-coordinate of $s^t_q$ is in between the $y$-coordinates of $s^t_p$ and $s^t_r$. Let $\Delta$ be the triangle determined by the segments $s_p,s_r$ and the ground line. Note that $s_q$ lies between $s_p$ and $s_r$ on the ground line. Since $s_q$ neither intersects $s_p$ nor $s_r$, it lies   interior to $\Delta$, which contradicts that the $y$-coordinate of $s^t_q$ is in between the $y$-coordinates of $s^t_p$ and $s^t_r$.
\end{proof}

\section{Details of Section~\ref{sec:unit}}
\label{app:cocom}

\begin{remark}
\label{rem:trivial}
    Let $\Delta abc$ be a triangle with base $ab$ and peak $c$. Let $q$ be a point on the segment $ab$. Then $|qc|\le \max\{|ac|,|bc|\}$.
\end{remark}

\begin{proof}[Proof of Lemma~\ref{slide}]
Let $\sigma=(s_\ell,\ldots,s_r)$ be the order of the segments on the ground line. Note that by definition, all the segments in $G(s_\ell ,s_r)$ intersect both $s_\ell$ and $s_r$. We now show that $\sigma$ is an umbrella-free ordering for $G(s_\ell ,s_r)$, and thus $G(s_\ell ,s_r)$ is a cocomparability graph by Lemma~\ref{lem:umbrella}. Specifically, let $s_a,s_c,s_b$ be three segments that appear in $\sigma$ in this order. We now show that if $s_a$ and $s_b$ intersect, then at least one of them must intersect $s_c$. 

 For a segment $s_a$, we refer to its endpoint on the ground line as $s^g_a$ and the free endpoint   as $s^t_a$. If a pair of segments $s_a$ and  $s_b$ intersect, then we denote the intersection point as $i(s_a,s_b)$. Let $W$ be the closed wedge-shaped region above the the two lines determined by $s_\ell$ and $s_r$. We distinguish two cases depending on whether $i(s_a,s_b)$ lies in $W$. 


{\bf Case 1 ($i(s_a,s_b)\in W$)
.} Suppose for a contradiction that $s_c$ does not intersect $s_a$ and $s_b$. Since $s_c$ lies between $s_a$ and $s_b$,  the point $s^t_c$ must be inside the triangle $\Delta s^g_a i(s_a,s_b)s^g_b$. Since $s_c$ intersects both $s_\ell$ and $s_r$, $s^t_c$ belongs to $W$. 
Without loss of generality assume that extending the segment $s_c$ upward intersects $s_b$ at a point $w$ (Figure~\ref{fig:gr}(a)). We now  move $w$ along $s_b$ so that the length $|ws^g_c|$ increases. If $w$ moves  upward along $s_b$, then it reaches $i(s_a,s_b)$. We now apply Remark~\ref{rem:trivial} on $\Delta s^g_ai(s_a,s_b)s^g_b$ to observe that $|ws^g_c|$ is at most one unit. Since the length of $s_c$ is strictly smaller than $|ws^g_c|$, $s_c$ cannot be of unit length. If $w$ moves  downward along $s_b$, then it reaches $i(s_\ell,s_b)$. Note that we may also hit $i(s_r,s_b)$ while moving $w$, but we keep moving until we reach $i(s_\ell,s_b)$  (Figure~\ref{fig:gr}(b)). We now can use the same argument by applying  Remark~\ref{rem:trivial} on $\Delta s^g_\ell i(s_\ell,s_b)s^g_b$.

{\bf Case 2 ($i(s_a,s_b)\not\in W$)
.} If $i(s_a,s_b)\in \Delta s^g_\ell i(s_\ell, s_r) s^g_r$, then this triangle contains $\Delta s^g_a i(s_a, s_b) s^g_b$ (Figure~\ref{fig:gr}(c)). Since $s_c$ lies between $s_a$ and $s_b$, $s_c$ cannot reach the segments $s_\ell$ and $s_r$.  If $i(s_a,s_b)\not \in \Delta s^g_\ell i(s_\ell, s_r) s^g_r$, then without loss of generality assume that $i(s_a,s_b)$ lies below the line determined by $s_\ell$ (Figure~\ref{fig:gr}(d)). Since $\Delta s^g_a i(s_a, s_b) s^g_b$ lies below  the line determined by $s_\ell$, $s_c$ cannot intersect $s_\ell$.
\end{proof}

\section{Details of Section~\ref{sec:gd}}
\label{app:gd}
\begin{proof}[Proof of Theorem~\ref{thm:polydisk}]
For each disk $D\in\mathcal{D}$, we compute a maximum clique $M(D)$ that contains  $D$ as its smallest disk. A maximum clique of $G(\mathcal{D})$ is then obtained by taking the largest among these cliques. 
To compute $M(D)$, we can only examine the set of disks  $S\subseteq \mathcal{D}$ that intersect $D$ (Figure~\ref{fig:gd}(a)). It now suffices to show that the intersection graph of $S$ is   a cobipartite graph. 
Let $S_\ell\subset S$ be the disks such that for each disk $D' \in S$, the $x$-coordinate of $D'$ is at most that of the center of $D$. Let $S_r$ to be the set of  remaining disks of $S$. We now show that the disks in $S_\ell$ mutually intersects, and By symmetry, the same argument applies to $S_r$. 

\begin{figure}[h]
    \centering
    \includegraphics[width=\linewidth]{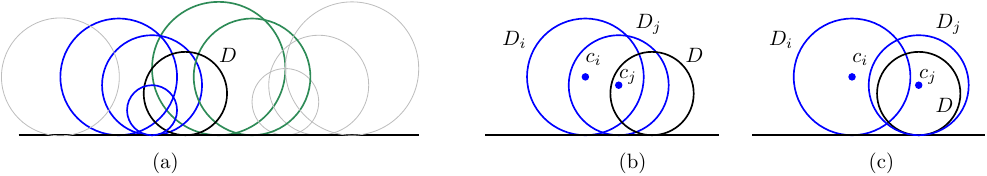}
    \caption{Illustration for Theorem~\ref{thm:polydisk}. (a) $\mathcal{D}$, where $D,S_\ell$ and $S_r$ are shown in black, blue and green, respectively. (b)--(c) Illustration for moving $D_j$ to contain $D$. }
    \label{fig:gd}
\end{figure}
Let $D_i, D_j \in S_\ell$ be two disks with radii $r_i$ and $r_j$, and centers $c_i$ and $c_j$, respectively. Let $d$ be the Euclidean distance between $c_i$ and $c_j$.  
Assume without loss of generality that the $x$-coordinate of $c_i$ is at most the $x$-coordinate of $c_j$. We  move $D_j$ horizontally such that the $x$-coordinate of $c_j$ becomes equal to that of $D$. Let $d'$ be the new distance between $D_i$ and $D_j$. We  now observe that  $d'\ge d$. Since $D$ is the smallest disk and both $D_j$ and $D$ are grounded, $D_j$ now entirely covers $D$ (Figures~\ref{fig:gd}(b)--(c)). Since $D_i$ intersects $D$, it must intersect $D_j$, i.e., $d'\le r_i+r_j$. Since $d\le d'$, $D_i$ must intersect $D_j$ in their original positions. 
\end{proof}

\section{Details of Section~\ref{sec:piercing}}
\label{app:piercing}

\begin{proof}[Proof of Lemma~\ref{lem:contain}] The proof follows since every point $t$ on $pc_2$ is closer to $c_1$ than $c_2$.
\end{proof}

\begin{proof}[Proof of Lemma~\ref{lem:arc}]
 Let $R$ be the region of $D_1$ that lies outside $D_2$. Since two distinct circles   can intersect in at most two points, $R$ is connected. Let $C'$ be the largest arc on the boundary of $D_1$ that also lies on $R$. Since $r\not \in D_2$, we have $r\in R$ and hence the arc $C$ coincides with $C'$.   
\end{proof}

\begin{proof}[Proof of Lemma~\ref{lem:120}]
Suppose for a contradiction that $D_1$ and $D_2$ intersect. Let $c_i $ be the center  of $D_i$, where $i\in\{1,2\}$. Let $D$ be a disk containing $D_1$ that touches $p$ and $M$ and has radius three. Since $D_1$ and $D_2$ intersect and both avoid $o$, $D$ must intersect $D_2$. Therefore, we  assume without loss of generality that $D_1$ is of radius three, and similarly $D_2$ is of radius three. 

 Note that $D_1$ and $D_2$ both avoid containing $o$, and $D_1$ touches $p$ and $M$. If $c_1$ lies above $L(c_2p)$, then the part of $D_1$ below $L(c_2p)$ lies below the line that passes through $p$ and perpendicular to $L(pc_2)$. 
 Since we assume $D_1$ intersects $D_2$ below the line $y=- \frac{x}{\sqrt{3}}$, $c_1$ must lie on or below $L(c_2p)$. 
 Figure~\ref{fig:lem120} illustrates the configuration. Increasing the radius of $M$ keeps $c_2$ on line $y=-x/\sqrt{3}$, but moves $c_1$ further away from $o$, which keeps $c_1$ below $L(c_2p)$. We now rotate $D_1$ clockwise such that it touches $p$ and $|c_1o|$ becomes $3(\frac{2}{\sqrt{3}}-1)$. Let $c_{old}$ be the position of $c_1$ before the rotation. Let $Z$ be a circle with center $c_2$ and radius $c_{old}$. Since $p$ lies inside $Z$, and  on or above $L(c_{old}c_2)$, $c_1$ lies inside $Z$. Therefore, $D_1$ still intersects $D_2$ after rotation.  
\begin{figure}
    \centering
    \includegraphics[width=0.7\linewidth]{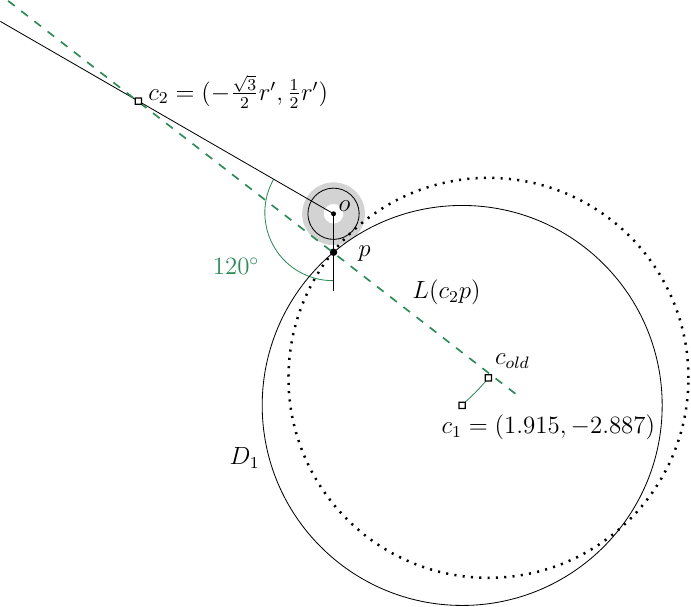}
    \caption{Illustration for Lemma~\ref{lem:120}. The gray annulus illustrates the range for $M$. The dotted disk is the position of $D_1$ before rotation.}
    \label{fig:lem120}
\end{figure}

We now have $c_1$ at $(1.915,-2.887)$, which can be obtained from the equations $x^2+(y+\tfrac{1}{\sqrt{3}})^2 = 9$ and $x^2+y^2 = r'^2$, where $r' = 3+r(\frac{2}{\sqrt{3}}-1)$. The center $c_2$ is at $(-\tfrac{\sqrt{3}}{2}r',\tfrac{1}{2}r')$. We now have $
    |c_1c_2| = \sqrt{(1.915+\tfrac{\sqrt{3}}{2}r')^2+( -2.887 - \tfrac{1}{2}r')^2}$. For $r\in [1,3]$, we have $|c_1c_2|\in  [6.44, 6.74]$, which contradicts the assumption that $D_1$ and $D_2$ intersects.  
\end{proof}

\subsection{Details of Case 2}\label{app:case2}

\smallskip
\noindent
{\bf Case 2.1 ($D'$ intersects $D_\ell$ below $L(oC)$).}  Let $r_\ell$ be the radius of $D_\ell$. We first translate $D'$ such that it touches $A$ and $Q$ and also maintains the intersection with $D_\ell$, as follows.  We translate $D'$ along $oc'$ until it touches $A$, which keeps the intersection with $Q$. By Lemma~\ref{lem:contain}, $D'$ continues to intersect $D_\ell$. Since $c_\ell$ in the second quadrant (Lemma~\ref{lem:tan}) and $D_\ell$ avoids $o$, $D_\ell$ cannot contain $A$. Since $D_\ell$ and $D'$ both avoid $A$, and since $D_\ell$ and $D'$ intersect below $L(oC)$, $A$ is above $L(c'c_\ell)$. 
We now rotate $D'$ clockwise such that the distance $|c'A|$ remains fixed but  $c'$ touches $Q$ (Figure~\ref{fig:transf2}(a)). Let $c'_{old}$ be the position of $c'$ before rotation. Consider a point $A'$ where $|c'_{old}A'| = |c'A'| = |c_\ell c'|$ (Figure~\ref{fig:transf2}(b)). Since  $A$ lies above $L(c_\ell c'_{old})$, $A'$ lies above $L(c_\ell c'_{old})$. Moving $A'$ to $c_\ell$ without changing the distance $|A'c'_{old}|$ decreases the angle $\angle c'c'_{old}A'$. We now consider the angles of $\Delta c' A'c_{old}$ to observe that 
$|c_\ell c'|\le |c_\ell c_{old}|$. Therefore, after rotation,  $D'$ continues to intersect  $D_\ell$.  
\begin{figure}[h]
    \centering
    \includegraphics[width=0.7\linewidth]{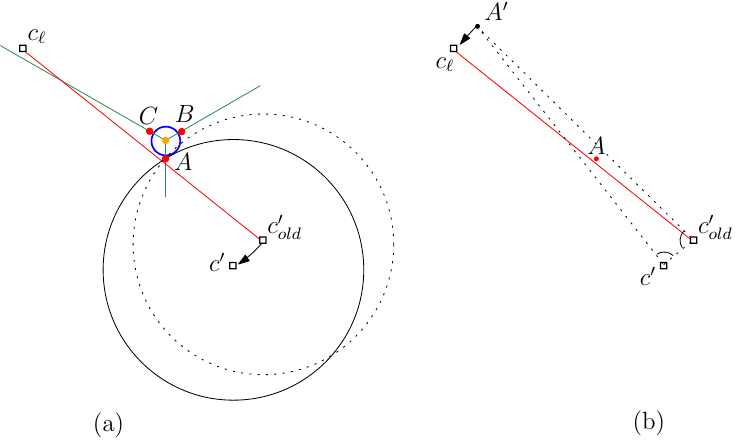}
    \caption{Illustration for moving $D'$  to touch $A$ and $Q$. The origin $o$ is in orange.}
    \label{fig:transf2}
\end{figure}

 We now define a disk $D''$ with center $c''$ and radius three that touches $A$ and has its center on $L(oc')$. Since $D''$ contains the part of $D'$ to the left of $L(oA)$, $D''$ intersects $D_\ell$ (Figure~\ref{fig:transf3}(a)).  
Since the largest radius among the disks of $\mathcal{D}$ is three, the radius of $Q$ is at most $3(\frac{2}{\sqrt{3}}-1)$ (Lemma~\ref{lem:hradi}).  
We now rotate $D''$ clockwise around $A$ such that the distance $|oc''|$ becomes $3+3(\frac{2}{\sqrt{3}}-1)$. We can now use the same argument that we used for rotating $D'$ to observe that $D''$ continues to intersect $D_\ell$. Since $(\frac{2}{\sqrt{3}}-1)\le r_q\le 3(\frac{2}{\sqrt{3}}-1)\approx 0.46$ and since $D''$ touches $A$, one can observe that $c''$ lies in the fourth quadrant.

We now consider two scenarios based on the position of $c_\ell$. 
Assume first that the center of $c_\ell$ lies above $L(oC)$. We roll $D_\ell$ over $Q$ counter-clockwise such that its center $c_\ell$ lies on  $L(oC)$. Since $c_\ell$ lies at the second quadrant (Lemma~\ref{lem:tan}), rolling over $Q$ decreases the $y$-coordinate of $c_\ell$ (Figure~\ref{fig:transf3}(b)). Let $c_{old}$ be the position of $c_\ell$ before moving $D_\ell$. Then $\angle c_\ell c_{old}c'' \le \angle c'' c_\ell  c_{old}$, and hence  $D_\ell$ continue to intersect $D'$. We now apply Lemma~\ref{lem:120} considering $D''$ and $D_\ell$ as $D_1$ and $D_2$ to reach a contradiction.

Assume  now that  $c_\ell$ lies below $L(oC)$.  Lemma~\ref{lem:tan} implies that  $D_\ell$ touches the \emph{upper-left arc} of $D_b$, i.e., the clockwise arc on the boundary from the leftmost point to the topmost point of $D_b$.   We now show that the upper-left arc of $D_b$ is outside of $D''$, which contradicts that $D''$ intersects $D_\ell$ (Figure~\ref{fig:transf3}(c)).

\begin{figure}[h]
    \centering
    \includegraphics[width=.9\linewidth]{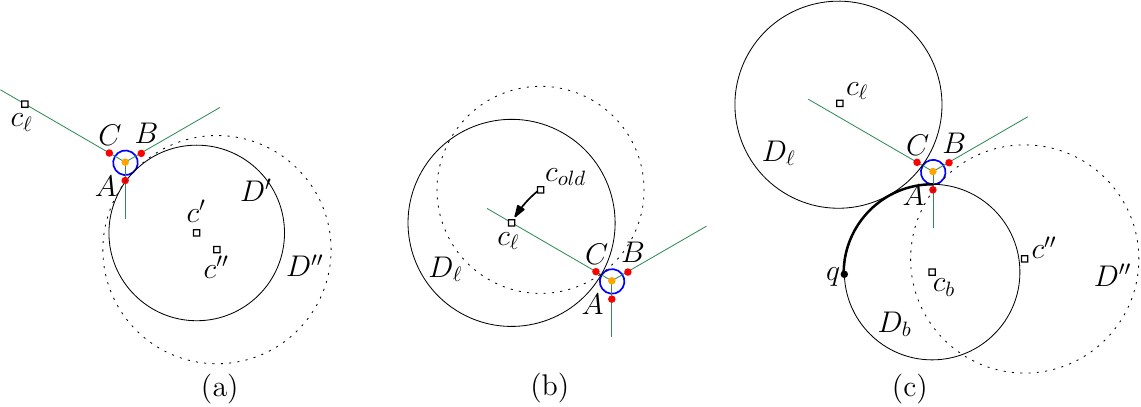}
    \caption{The origin  $o$ is in orange. Illustration showing (a) $D''$ (b) $D_\ell$ rolls over $Q$, and (c) $q$ lies outside of $D''$. }
    \label{fig:transf3}
\end{figure}

We first compute the center $c''=(x,y)$ of $D''$.  Since $c''$ is at distance $3+3(\frac{2}{\sqrt{3}}-1)=2\sqrt{3}$ from $o$, we have $x^2+y^2=(2\sqrt{3})^2$. Furthermore, $c''$ is at distance $3$ from $A$, and hence $x^2+(y+\frac{1}{\sqrt{3}})^2=3^2$. Consequently, $c''\approx (1.915,-2.887)$. Let $r_b$ be the radius of $D_b$. The distance between $c''$ and the leftmost point $q=(-r_b,-r_b-\frac{1}{\sqrt{3}})$ of $D_b$  is $|c''q|\approx \sqrt{(r_b+1.915)^2 + (r_b-2.310)^2 }$, 
which is larger than 3.19 for every $r_b\ge 1$. Therefore, $D''$ cannot contain $q$. By Lemma~\ref{lem:arc}, the upper-left arc of $D_b$ lies outside of $D''$, as required. 

\smallskip 
\noindent
{\bf Case 2.2 ($D'$ intersects $D_\ell$ above $L(oC)$).} 
If $c_\ell$ lies below $L(oC)$,  then we use the same argument of Case 2.1 when $c_\ell$ is above $L(oC)$. Specifically, we move $D'$ to touch $B$ and $Q$, and roll $D_\ell$ over $Q$ such that $c_\ell$ lies on $L(oC)$ to apply the arguments of Case 2.1 (Figure~\ref{fig:transf4}(a)).

We now assume that  $c_\ell$ lies above $L(oC)$ and $c'\in W(\angle BoA)$. 
Consider the triangle $\Delta c_b o c_\ell$  where $|c_\ell o| = r_\ell + r_q $, $|c_b o| = r_b + r_q $ and $|c_\ell c_b| = r_\ell + r_b $ (Figure~\ref{fig:transf4}(b)). The angle  $\angle c_b o c_\ell  = \arccos \left( \frac{(r_\ell+r_q)^2 + (r_b+r_q)^2 -(r_\ell+ r_b )^2 }{2(r_\ell+r_q)(r_b+r_q)}\right) = \arccos \left( \frac{r_\ell r_q +  r_b r_q + r_q^2 -r_\ell r_b    }{(r_\ell+r_q)(r_b+r_q)}\right)$ $= \arccos \left( 1-\frac{2r_\ell r_b   }{(r_\ell+r_q)(r_b+r_q)}\right)$.  
 For a fixed $r_\ell,r_b$, the angle increases when $r_q$ decreases. For fixed $r_q$, the angle increases when $r_\ell$ and $r_b$ increase. Therefore, $\angle c_b o c_\ell$ is maximized when    when $r_q =  (\frac{2}{\sqrt{3}}-1)$ and $r_\ell=r_b=3$, and hence $\angle c_b o c_\ell \approx 143.97^\circ $.

Similar to Case 2.1, we show the existence of  a disk $D''$ with center $c''$ and radius three that touches $B$, has  distance $|oc''| = 3+r_q$ and intersects $D_\ell$. If $r_q=3(\frac{2}{\sqrt{3}}-1)$, then $c''$ is approximately $(3.457,-0.215)$ with $L(oc'')$ making an angle $-3.56^\circ$ with positive x-axis. Hence it suffices to consider the coordinate of $c'' = (|oc''|\cos \theta, -|oc''|\sin \theta)$, where $0^\circ \le \theta \le 90^\circ$.
 \begin{figure}[h]
    \centering
    \includegraphics[width=.85\linewidth]{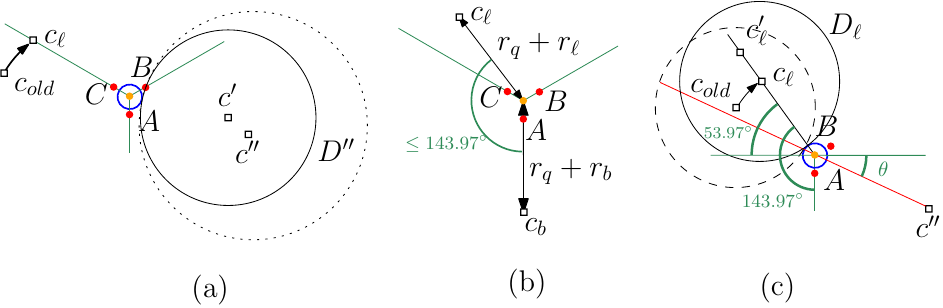}
    \caption{Illustration for Case 2.2. (a) The setup when $c_\ell$ lies below $L(oC)$. (b)--(c) The center $c_\ell$ lies above $L(oC)$.}
    \label{fig:transf4}
\end{figure}

We now roll $D_\ell$ over $Q$ clockwise  such that $c_\ell$ lies on the line $L$ with angle of inclination $143.97^\circ$ (Figure~\ref{fig:transf4}(c)). Since $D_\ell$ intersects  $D''$ above $L(oC)$,  $c_\ell$ must lie above $L(oc'')$. Hence $D_\ell$ continues to intersect $D''$. Let $D'_\ell$ be a disk of radius three tangent to $Q$ at the point where   $D_\ell$ touches $Q$. Since $D'_\ell$ contains $D_\ell$, it intersects $D''$. Since the center $c'_\ell$ of   $D'_\ell$ lies on $L(oc_\ell)$,  $\angle c_b o c'_\ell \le 143.97^\circ $. Therefore, we may assume that the center $c'_\ell$ has the coordinate $ (-r \cos 53.97^\circ , r \sin 53.97^\circ )$, where $r= 3 + r_q$ (Figure~\ref{fig:transf4}(c)).

We now reach a contradiction by showing $|c'_\ell c''|$ is larger than six, as follows.
\begin{align*}
|c'_\ell c''| &=  \sqrt{(-r \cos 53.97^\circ-|oc''|\cos \theta)^2 + ( r \sin 53.97^\circ+|oc''| \sin\theta )^2} \\
&= \sqrt{r^2+|oc''|^2+2r|oc''| \cos 53.97^\circ \cos \theta + 2r|oc''| \sin 53.97^\circ \sin\theta}\\
&=\sqrt{r^2+|oc''|^2+2r|oc''|  \cos (53.97^\circ - \theta) }\\
& = \sqrt{2r^2+2r^2\cos (53.97^\circ - \theta) }\\
\end{align*}


If $r_q= 2(\frac{2}{\sqrt{3}}-1)$, then $c''\approx (3.039,-1.308)$  with $L(oc'')$ making an angle $-23.3^\circ$ with positive x-axis. Therefore, if $ (\frac{2}{\sqrt{3}}-1)\le r_q<  2(\frac{2}{\sqrt{3}}-1)$, then $23.3^\circ \le \theta \le 90^\circ$. The distance $|c'_\ell c''|$ is minimized when $\theta=90^\circ$. Since distance decreases as $r_q$ decreases,  we have $|c'_\ell c''|\ge \sqrt{2 (3+(\frac{2}{\sqrt{3}}-1))^2 (1+\cos 36.03^\circ) } >6$.

If $ 2(\frac{2}{\sqrt{3}}-1)\le r_q\le  3(\frac{2}{\sqrt{3}}-1)$, then $0^\circ \le \theta \le 23.3^\circ$. The distance $|c'_\ell c''|$ is minimized when $\theta=23.3^\circ$. Since distance decreases as $r_q$ decreases,  we have $|c'_\ell c''|\ge \sqrt{2 (3+2(\frac{2}{\sqrt{3}}-1))^2 (1+\cos 23.3^\circ) } >6$.





\end{document}